\documentclass[
 onecolumn,
 11pt, 
 draftcls
]{IEEEtran}

\usepackage{array}
\usepackage{epsfig}
\usepackage{cite}
\usepackage{amsfonts}
\usepackage{amsmath}
\usepackage{amssymb}
\usepackage{amsthm}
\usepackage{amsxtra}
\usepackage{varioref}
\usepackage{multicol}
\usepackage{float}
\usepackage{rotate}
\usepackage{color}
\usepackage{enumerate}
\usepackage[active]{srcltx} 
\usepackage{mathrsfs}
\usepackage[usenames,dvipsnames,svgnames,table]{xcolor}
\usepackage{rotating}
\usepackage{Milancommands}

\definecolor{milancolor}{RGB}{160,80,0}
\definecolor{matiascolor}{RGB}{0,0,153}

\newcommand{\rows}[3][\bQ_{n}]{[#1]^{#2}_{#3}}
\newcommand{\block}[3][\bQ_{n}]{{^{#2}\![#1]_{#3}}}
\newcommand{\bSigma}{\boldsymbol{\Sigma}}

\usepackage{times}
\usepackage[english]{babel}
\usepackage[utf8]{inputenc}
\usepackage{fontenc}
\usepackage{mathtext}
\usepackage[colorlinks=true,linkcolor=black,citecolor=black]{hyperref}

\newcommand{\Dmn}{{^{m+1}\![\bD_{n}]_{n}}}
\newcommand{\bPhi}{\boldsymbol{\Phi}}
\newcommand{\bPsi}{\boldsymbol{\Psi}}
\newcommand{\bTheta}{\boldsymbol{\Theta}}

\newtheorem{thm}{Theorem}
\newtheorem{defn}{Definition}
\newtheorem{coro}{Corollary}
\newtheorem{lem}{Lemma}
\newtheorem{rem}{Remark}
\newtheorem{prop}{Proposition}

\newtheorem{assu}{\textbf{Assumption}}

\title{The Entropy Gain of Linear Time-Invariant Filters and Some of its Implications
} 

\author{Milan S. Derpich, 
Matías Müller\thanks{
Milan S. Derpich and Matías Müller are with the Department of Electronic Engineering, Universidad Técnica Federico Santa María, Chile.
Emails: milan.derpich@usm.cl,matias.muller@alumnos.usm.cl. 
Jan {\O}stergaard is with the
Department of Electronic Systems,
Aalborg University, Denmark.
Email: jo@es.aau.dk.
This work is partially supported by CONICYT PCHA/ Beca Magíster Complementario/Año 2013 – Folio 221320226, CONICYT Fondecyt grant 1140384 and CONICYT Basal Fund FB0008.} 
and Jan {\O}stergaard}

\begin{document}

\maketitle
\begin{abstract}
We study the increase in per-sample differential entropy rate of random sequences and processes after being passed through
a non minimum-phase (NMP) discrete-time, linear time-invariant (LTI) filter $G$.
For LTI discrete-time filters and random processes, it has long been established that this entropy gain, $\Gsp(G)$, equals the integral of $\log\abs{G\ejw}$.
It is also known that, if the first sample of the impulse response of $G$ has unit-magnitude, then the latter integral 
equals the sum of the logarithm of the magnitudes of the non-minimum phase zeros of $G$ (i.e., its zeros outside the unit circle), say $\Bsp(G)$.
These existing results have been derived in the frequency domain as well as in the time domain.
In this note, we begin by showing that existing time-domain proofs, which consider finite length-$n$ sequences and then let $n$ tend to infinity, have neglected significant mathematical terms and, therefore, are inaccurate.
We discuss some of the implications of this oversight when considering random processes.
We then present a rigorous time-domain analysis of the entropy gain of LTI filters for random processes.
In particular, we show that the entropy gain between equal-length input and output sequences is upper bounded by $\Bsp(G)$ and arises if and only if there exists an output additive disturbance with finite differential entropy (no matter how small) or a random initial state.
Unlike what happens with linear maps, the entropy gain in this case depends on the distribution of all the signals involved.
Instead, when comparing the input differential entropy to that of the entire (longer) output of $G$, the entropy gain equals $\Bsp(G)$ irrespective of the distributions and without the need for additional exogenous random signals.
We illustrate some of the consequences of these results by presenting their implications in three different problems. Specifically: 
a simple derivation of the rate-distortion function for Gaussian non-stationary sources, 
conditions for equality in an information inequality of importance in networked control problems,
and an observation on the capacity of auto-regressive Gaussian channels with feedback.
\end{abstract}

\section{Introduction}\label{sec:Intro}
In his seminal 1948 paper~\cite{shawea49}, Claude Shannon gave a formula for the increase in differential entropy per degree of freedom that a continuous-time, band-limited random process $\rvau(t)$ experiences after passing through a linear time-invariant (LTI) continuous-time filter.
In this formula, if the input process is band-limited to a frequency range $[0,B]$, has differential entropy rate (per degree of freedom) $\bar{h}(\rvau)$, and the LTI filter has frequency response $G\jw$, then the resulting differential entropy rate of the output process $\rvay(t)$  is given by%
\cite[Theorem~14]{shawea49}
\begin{align}\label{eq:hgain_Shannon}
 \bar{h}(\rvay) = \bar{h}(\rvau)  + \frac{2}{B}\Intfromto{0}{B} \log \abs{G\jw}d\w.
\end{align}
The last term on the right-hand side (RHS) of~\eqref{eq:hgain_Shannon} 
can be understood as the \textit{entropy gain} (entropy amplification or entropy boost)
introduced by the filter~$G\jw$.
Shannon proved this result by arguing that an LTI filter can be seen as a linear operator that selectively scales its input signal 
along infinitely many frequencies, each of them representing an orthogonal component of the source.
The result is then obtained by writing down the determinant of the Jacobian of this operator as the product of the frequency response of the filter over $n$ frequency bands, applying logarithm and then taking the limit as the number of frequency components tends to infinity.

An analogous result can be obtained for discrete-time input $\procu$ and output $\procy$ processes,  and an LTI discrete-time filter $G(z)$ by relating them to their continuous-time counterparts, which yields
\begin{align}\label{eq:hgain_discrete}
 \bar{h}(\procy) = \bar{h}(\procu) + \frac{1}{2\pi}\intfromto{-\pi}{\pi}\log\abs{G\ejw}d\w,
\end{align}
where 
$$
\bar{h}(\procu)\eq \lim_{n\to\infty} \tfrac{1}{n} h(\rvau(1),\rvau(2),\ldots,\rvau(n))
$$
is the differential entropy rate of the process $\procu$.
Of course the same formula can also be obtained by applying the frequency-domain proof technique that Shannon followed in his derivation of~\eqref{eq:hgain_Shannon}.

The rightmost term in~\eqref{eq:hgain_discrete}, which corresponds to the entropy gain of $G(z)$, can be related to the structure of this filter. 
It is well known that if $G$ is causal with a rational transfer function $G(z)$ such that 
$\lim_{z\to\infty}|G(z)|=1$ (i.e., such that the first sample of its impulse response has unit magnitude), then  
\begin{align}\label{eq:Jensen}
\frac{1}{2\pi}\intfromto{-\pi}{\pi}\log\abs{G\ejw}d\w 
 = \Sumover{c_{i}\notin\mathbb{D}}\log\abs{\rho_{i}},  
\end{align}
where $\set{\rho_{i}}$ are the zeros of $G(z)$ and 
${\mathbb{D}}\eq \set{z\in\mathbb{C}: \abs{z} <1}$ is the open unit disk on the complex plane.
This provides a straightforward way to evaluate the entropy gain of a given LTI filter with rational transfer function $G(z)$.
In addition,~\eqref{eq:Jensen} shows that, if $\lim_{z\to\infty}|G(z)|=1$, then such gain is greater than one if and only if $G(z)$ has zeros outside $\mathbb{D}$.
A filter with the latter property is said to be \textit{non-minimum phase} (NMP); conversely, a filter with all its zeros inside $\mathbb{D}$ is said to be \textit{minimum phase} (MP)~\cite{serbra97}.

NMP filters appear naturally in various applications. 
For instance, any unstable LTI system stabilized via linear feedback control will yield transfer functions which are NMP~\cite{serbra97,googra00}.
Additionally, NMP-zeros also appear when a discrete-time with ZOH (\emph{zero order hold}) equivalent system is obtained from a plant whose number of poles exceeds its number of zeros by at least 2, as the sampling rate increases~\cite[Lemma 5.2]{yuzgoo14}. 
On the other hand, all linear-phase filters, 
which are specially suited for audio and image-processing applications, are NMP~\cite{hayes-96,vaidya93}. 
The same is true for any all-pass filter, which is an important building block in signal processing applications~\cite{smith-07,hayes-96}.


An alternative approach for obtaining the entropy gain of LTI filters is to work in the time domain; obtain $\rvay_{1}^{n}\eq \set{\rvay_{1},\rvay_{1},\ldots,\rvay_{n}}$
as a function of $\rvau_{1}^{n}$, for every $n\in\Nl$,
and evaluate the limit 
$
 \lim_{n\to\infty}\frac{1}{n}\left(h(\rvay_{1}^{n}) - h(\rvau_{1}^{n}) \right)
$.
More precisely, for a filter $G$ with impulse response $g_{0}^{\infty}$, we can write
\begin{align}\label{eq:y_of_u_matrix}
 \rvey^{1}_{n} = 
\underbrace{\begin{pmatrix}
 g_{0} & 0    &\cdots& 0\\
 g_{1} & g_{0}&\cdots &0\\
 \vdots & &\ddots &\vdots\\
 g_{n-1}& g_{n-2}&\cdots &g_{0}
\end{pmatrix}}_{\bG_{n}}
\rveu^{1}_{n},
\end{align}
where $\rvey^{1}_{n} \eq [\rvay_{1}\ \rvay_{1}\,\cdots \ \rvay_{n}]^{T}$ and the random  vector $\rveu^{1}_{n}$ is defined likewise. 
From this, it is clear that 
\begin{align}\label{eq:hy=hu+logdet}
 h(\rvey^{1}_{n}) = h(\rveu^{1}_{n}) + \log|\det(\bG_{n})|,
\end{align}
where $\det(\bG_n)$ (or simply $\det \bG_n$) stands for the determinant of $\bG_n$. Thus, 
%
%
%
\begin{align}\label{eq:f0=1_andtherest}
|g_{0}|=1\Longrightarrow |\det(\bG_{n})|=1, \;\forall n\in\Nl \Longleftrightarrow
h(\rvey^{1}_{n}) = h(\rveu^{1}_{n}),\;\forall n\in\Nl
\Longrightarrow
 \lim_{n\to\infty}\frac{1}{n}\left[h(\rvey_{1}^{n}) - h(\rveu_{1}^{n})\right] =0,
\end{align}
regardless of whether $G(z)$ (i.e., the polynomial $g_{0} + g_{1}z^{-1}+\cdots $) has zeros with magnitude greater than one, \textbf{which clearly
contradicts~\eqref{eq:hgain_discrete} and~\eqref{eq:Jensen}}. 
Perhaps surprisingly, the above contradiction not only has been overlooked in previous works (such as~\cite{aarmcd67,zanigl03}), but the time-domain formulation  in the form of~\eqref{eq:y_of_u_matrix} has been utilized as a means to prove or disprove~\eqref{eq:hgain_discrete} (see, for example, the reasoning in~\cite[p.~568]{papou91}).

A reason for why the contradiction between~\eqref{eq:hgain_discrete},~\eqref{eq:Jensen} and~\eqref{eq:f0=1_andtherest} arises 
can be obtained from the analysis developed in~\cite{mardah08} for an LTI system $P$ within a noisy feedback loop, as the one depicted in Fig.~\ref{fig:fbksystem}.
In this scheme, $C$ represents a causal feedback channel which combines the output of $P$ with an exogenous (noise) random process $\rvac_{1}^{\infty}$ to generate its output.
The process $\rvac_{1}^{\infty}$ is assumed independent of the initial state of $P$, represented by the random vector $\rvex_{0}$, which has finite differential entropy.
\begin{figure}[t]
\centering
\input{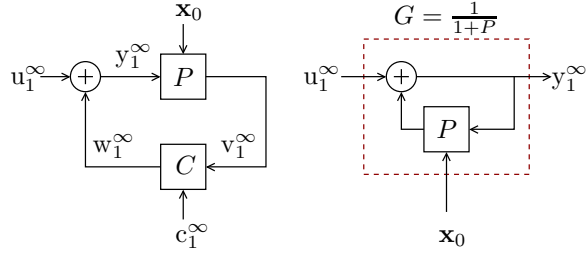}
\caption{Left: LTI system $P$ within a noisy feedback loop. Right: equivalent system when the feedback channel is noiseless and has unit gain.}
\label{fig:fbksystem}
\end{figure}
For this system, it is shown in~\cite[Theorem 4.2]{mardah08} that 
\begin{subequations}\label{eq:martins_both}
 \begin{align}\label{eq:martins}
 \bar{h}(\rvay_{1}^{\infty}) \geq \bar{h}(\rvau_{1}^{\infty}) + \lim_{n\to\infty}\frac{1}{n}I(\rvex_{0}; \rvay_{1}^{n}),
\end{align}
with equality if $\rvaw$ is a deterministic function of $\rvav$.
Furthermore, it is shown in~\cite[Lemma 3.2]{mardah05} that if $|h(\rvex_{0})|<\infty$ and the steady state variance of system $P$ remains asymptotically bounded as $k\to\infty$, then  
\begin{align}\label{eq:martins_I_bound}
\lim_{n\to\infty}\frac{1}{n}I(\rvex_{0}; \rvay_{1}^{n})
\geq \sumover{p_{i}\notin\mathbb{D}}\log\abs{p_{i}},
\end{align}
\end{subequations}
where $\set{p_{i}}$ are the poles of $P$.
Thus, for the (simplest) case in which $\rvaw=\rvav$, the output $\rvay_{1}^{\infty}$ is the result of filtering $\rvau_{1}^{\infty}$ by a filter $G=\frac{1}{1+P}$ (as shown in Fig.~\ref{fig:fbksystem}-right), and the resulting entropy rate of $\procy$ will exceed that of $\procu$ only if there is a random initial state with bounded differential entropy (see~\eqref{eq:martins}). 
Moreover, under the latter conditions,~\cite[Lemma 4.3]{mardah08} implies that if $G(z)$ is stable and $|h(\rvex_{0})|<\infty$, then this entropy gain will be lower bounded by the \textit{right-hand side} (RHS) of~\eqref{eq:Jensen}, which is greater than zero if and only if $G$ is NMP. 
However, the result obtained in~\eqref{eq:martins_I_bound} does not provide conditions under which the equality in the latter equation holds.

Additional results and intuition related to this problem can be obtained from in~\cite{kim-yh10}.
There it is shown that if $\procy$ is a two-sided Gaussian stationary random process generated by a state-space recursion of the form 
\begin{subequations}\label{subeq:State_Space_YHKIM}
\begin{align}
 \rves_{k+1}& = (\bA - \bg\bh^{H}) \rves_{k} - \bg \rvau_{k},\\
 \rvay_{k}  & = \bh^{H} \rves_{k} + \rvau_{n},
\end{align}
\end{subequations}
for some $\bA\in\mathbb{C}^{M\times M}$, $\bg\in\mathbb{C}^{M\times 1}$, $\bh\in\mathbb{C}^{M\times 1}$,
with unit-variance Gaussian i.i.d. innovations $\rvau_{-\infty}^{\infty}$, then its entropy rate will be exactly
$\frac{1}{2}\log(2\pi\expo{})$ (i.e., the differential entropy rate of $\rvau_{-\infty}^{\infty}$) plus the RHS of~\eqref{eq:Jensen} (with $\set{\rho_{i}}$ now being the eigenvalues of $\bA$ outside the unit circle).
However, as noted in~\cite{kim-yh10}, if the same system with zero (or deterministic) initial state is excited by a one-sided infinite Gaussian i.i.d. process $\rvau_{1}^{\infty}$ with unit sample variance, then the (asymptotic) entropy rate of the output process $\rvay_{1}^{\infty}$ is just~$\frac{1}{2}\log(2\pi\expo{})$ (i.e., there is no entropy gain).
Moreover, it is also shown that if $\rvav_{1}^{\ell}$ is a Gaussian random sequence with positive-definite covariance matrix and $\ell\geq M$, then the entropy rate of $\rvay_{1}^{\infty}+\rvav_{1}^{\ell}$ also exceeds that of $\rvau_{1}^{\infty}$ by the RHS of~\eqref{eq:Jensen}.
This suggests that for an LTI system which admits a state-space representation of the form~\eqref{subeq:State_Space_YHKIM}, the entropy gain 
for a single-sided Gaussian i.i.d. input is zero, and that the entropy gain from the input to the output-plus-disturbance is~\eqref{eq:Jensen}, for any Gaussian disturbance of length $M$ with positive definite covariance matrix (no matter how small this covariance matrix may be).


The previous analysis suggests that it is the absence of a random initial state or a random additive output disturbance that makes the time-domain formulation~\eqref{eq:y_of_u_matrix} yield a zero entropy gain.
But, how would the addition of such finite-energy exogenous random variables to~\eqref{eq:y_of_u_matrix} actually produce an increase in the differential entropy rate which asymptotically equals the RHS of~\eqref{eq:Jensen}? 
In a broader sense, it is not clear from the results mentioned above what the necessary and sufficient conditions 
are under which  
an entropy gain equal to the RHS of~\eqref{eq:Jensen} arises (the analysis in~\cite{kim-yh10} provides only a set of sufficient conditions and relies on second-order statistics and Gaussian innovations to derive the results previously described). 
Another important observation to be made is the following: it is well known that the entropy gain introduced by a linear mapping is independent of the input statistics~\cite{shawea49}. 
However, 
there is no reason to assume such independence 
when this entropy gain arises as the result of adding a random signal to the input of the mapping, i.e., when the mapping by itself does not produce the entropy gain. 
Hence, it remains to characterize the largest set of input statistics which yield an entropy gain, and the magnitude of this gain.

The first part of this paper provides answers to these questions. 
In particular, in Section~\ref{sec:geometric_interpretation} explain how and when the entropy gain arises (in the situations described above), starting with input and output sequences of finite length,
in a time-domain analysis similar to~\eqref{eq:y_of_u_matrix},
and then taking the limit as the length tends to infinity.
In Section~\ref{sec:entropy_gain_output_disturb} it is shown that, in the output-plus-disturbance scenario, 
the entropy gain is \emph{at most} the RHS of~\eqref{eq:Jensen}. 
We show that, for a broad class of input processes (not necessarily Gaussian or stationary), this maximum entropy gain
is reached only when the disturbance has bounded differential entropy and its length is at least equal to the number of non-minimum phase zeros of the filter.
We provide upper and lower bounds on the entropy gain if the latter condition is not met.
A similar result is shown to hold when there is a random initial state in the system (with finite differential entropy).
In addition, in Section~\ref{sec:entropy_gain_output_disturb} we study the entropy gain between the \emph{entire output sequence} that a filter yields as response to a shorter input sequence (in Section~\ref{sec:effective_entropy}).
In this case, however, it is necessary to consider a new definition for differential entropy, named \emph{effective differential entropy}.
Here we show that 
 an effective entropy gain equal to the RHS of~\eqref{eq:Jensen} is obtained provided the input has finite differential entropy rate, even when there is no random initial state or output disturbance.

In the second part of this paper (Section\ref{sec:implications}) we apply the conclusions obtained in the first part to three problems, namely, 
networked control,
the rate-distortion function for non-stationary Gaussian sources, 
and the Gaussian channel capacity with feedback.
In particular, we show that equality holds in~\eqref{eq:martins_I_bound} 
for the feedback system in Fig.~\ref{fig:fbksystem}-left 
under very general conditions (even when the channel $C$ is noisy).
For the problem of finding the quadratic rate-distortion function for non-stationary auto-regressive Gaussian sources, previously solved in~\cite{gray--70,hasari80,grahas08}, we provide a simpler proof based upon the results we derive in the first part.
This proof extends the result stated in~\cite{hasari80,grahas08} to a broader class of non-stationary sources. 
For the feedback Gaussian capacity problem, we show that capacity results based on 
using a short random sequence as channel input and relying on a feedback filter which boosts the entropy rate of the end-to-end channel noise (such as the one proposed in~\cite{kim-yh10}), crucially depend  upon the complete absence of any additional disturbance anywhere in the system.
Specifically, we show that the information rate of such capacity-achieving schemes drops to zero in the presence of any such additional disturbance.
As a consequence, the relevance of characterizing the robust (i.e., in the presence of disturbances) feedback capacity of Gaussian channels, which appears to be a fairly unexplored problem, becomes evident.

Finally, the main conclusions of this work are summarized in Section~\ref{sec:conclusions}.

Except where present, all proofs are presented in the appendix.

\subsection{Notation}
For any LTI system $G$, the transfer function $G(z)$ corresponds to the $z$-transform of the impulse response $g_{0},\, g_{1}, \ldots$, i.e., $G(z) = \sumfromto{i=0}{\infty} g_{i}z^{-i}$.
For a transfer function $G(z)$, we denote by $\bG_{n}\in\Rl^{n\times n}$ the lower triangular Toeplitz matrix having $[g_{0}\ \cdots \ g_{n-1}]^{T}$ as its first column.
We write $x_{1}^{n}$ as a shorthand for the sequence $\set{x_{1},\ldots, x_{n}}$ and, when convenient, we write $x_{1}^{n}$ in vector form as  $\bx^{1}_{n}\eq [x_{1}\ x_{2}\ \cdots \ x_{n}]^{T}$, where $()^{T}$ denotes transposition.
Random scalars  (vectors) are denoted using non-italic characters, such as $\rvax$ (non-italic and boldface characters, such as $\rvex$).
For matrices we use upper-case boldface symbols, such as $\bA$.
We write $\lambda_{i}(\bA)$
to the note the $i$-th smallest-magnitude eigenvalue of $\bA$. 
If $\bA_{n}\in\mathbb{C}^{n\times n}$, then $\bA_{i,j}$ denotes the entry in the intersection between the $i$-th row and the $j$-th column.
We write $[\bA_{n}]^{i_{1}}_{i_{2}}$, with $i_{1}\leq i_{2}\leq n$, to refer to the matrix formed by selecting the rows $i_{1}$ to $i_{2}$ of $\bA$.
The expression ${^{m_{1}}}\![\bA]_{m_{2}}$ corresponds to the square sub-matrix along the main diagonal of $\bA$, with its top-left and bottom-right corners on $\bA_{m_1,m_1}$ and $\bA_{m_{2},m_{2}}$, respectively.
A diagonal matrix whose entries are the elements in $\Dsp$ is denoted as $\diag\Dsp$

\section{Problem Definition and Assumptions}
Consider the discrete-time system depicted in Fig.~\ref{fig:general}.
In this setup, the input $\rvau_{1}^{\infty}$ is a random process
and the block $G$ is a causal, linear and time-invariant system with random initial state vector $\rvex_{0}$ and random output disturbance $\rvaz_{1}^{\infty}$.
\begin{figure}[t]
 \centering
\input{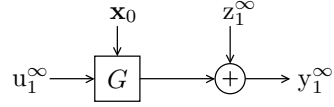}
\caption{Linear, causal, stable and time-invariant system $G$ with input and output processes, initial state and output disturbance.}
\label{fig:general}
\end{figure}
In vector notation,
\begin{align}\label{eq:system_equation}
 \rvey^{1}_{n}\eq \bG_{n}\rveu^{1}_{n} + \bar{\rvey}^1_n + \rvez^1_n,\fspace n\in\Nl,
\end{align}
where 
$\bar{\rvey}^1_n$
is the natural response of $G$ to the initial state $\rvex_{0}$.
We make the following further assumptions about $G$ and the signals around it:
\begin{assu}\label{assu:G_Factorized}
$G(z)$ is a causal, stable and rational transfer function of finite order, whose impulse response $g_{0},g_{1},\ldots$ satisfies
$ g_{0}=1$.
\finenunciado
\end{assu}
It is worth noting that there is no loss of generality in considering $g_{0}=1$, since otherwise one can write $G(z)$ as $G'(z)=g_{0}\cdot G(z)/g_{0}$,
  and thus the entropy gain introduced by $G'(z)$ would be $\log g_{0}$ plus the entropy gain due to  $G(z)/g_{0}$, which has an impulse response where the first sample equals $1$.

%
\begin{assu}
The random initial state $\rvex_{0}
$ is independent of $\rvau_{1}^{\infty}$.
\end{assu}
\begin{assu}\label{assu:z}
The disturbance $\rvaz_{1}^{\infty}$ is independent of $\rvau_{1}^{\infty}$ and  belongs to a $\kappa$-dimensional linear subspace, for some finite $\kappa\in\Nl$.
This subspace is spanned by the $\kappa$ orthonormal columns of a matrix $\boldsymbol{\Phi}\in\Rl^{|\Nl|\times \kappa}$ (where $|\Nl|$ stands for the countably infinite size of $\Nl$), such that 
$|h(\bPhi^{T}\rvez^{1}_{\infty})|<\infty$.  
Equivalently, $\rvez^{1}_{\infty} = \bPhi \rves^{1}_{\kappa}$, where the random vector 
$\rves^{1}_{\kappa}\eq \bPhi^{T}\rvez^{1}_{\infty}$ has finite differential entropy and is independent of $\rveu^{1}_{\infty}$.
\end{assu}

As anticipated in the Introduction, we are interested in characterizing the entropy gain $\Gsp$ of $G$ in the presence (or absence) of the random inputs
$\rvau_{1}^{\infty},\rvex_{0},\rvaz_{1}^{\infty}$, denoted by 
\begin{align}\label{eq:Gsp_def}
\Gsp(G,\rvex_{0},\rvau_{1}^{\infty},\rvaz_{1}^{\infty})\eq \lim_{n\to\infty}
\frac{1}{n}
\left(h(\rvay_{1}^{n})- h(\rvau_{1}^{n})\right).
\end{align}
In the next section we provide geometrical insight into the behaviour of $\Gsp(G,\rvex_{0},\rvau_{1}^{\infty},\rvaz_{1}^{\infty})$ for the situation where there is a random output disturbance and no random initial state. 
A formal and precise treatment of this scenario is then presented in Section~\ref{sec:entropy_gain_output_disturb}. 
The other scenarios are considered in the subsequent sections.

%
%

%

\section{Geometric Interpretation}\label{sec:geometric_interpretation}
In this section we provide an intuitive geometric interpretation of how and when the entropy gain defined in~\eqref{eq:Gsp_def} arises. 
This understanding will justify the introduction of the notion of an entropy-balanced random process (in Definition~\ref{def:entropy_balanced} below), which will be shown to play a key role in this and in related problems. 

\subsection{An Illustrative Example}
Suppose for the moment that $G$ in Fig.~\ref{fig:general} is an FIR filter with impulse response $g_{0}=1,\ g_{1}=2,\ g_{i}=0,\, \forall i\geq 2$.
Notice that this choice yields $G(z)= (z-2)/z$, and thus $G(z)$ has one non-minimum phase zero, at $z=2$.
The associated matrix $\bG_{n}$ for $n=3$ is
$$
\bG_{3}=\begin{pmatrix}
 1&0&0\\
 2&1&0\\
0&2&1
\end{pmatrix},
$$ 
whose determinant is clearly one (indeed, all its eigenvalues are $1$).
Hence, as discussed in the introduction, 
$h(\bG_{3}\rveu^{1}_{3})=h(\rveu^{1}_{3})$, and thus $\bG_{3}$ (and $\bG_{n}$ in general) does not introduce an entropy gain by itself.
However, an interesting phenomenon becomes evident by looking at the singular-value decomposition (SVD) of $\bG_{3}$, given by 
$\bG_{3}=\bQ_{3}^{T}\bD_{3}\bR_{3}$, 
where $\bQ_{3}$ and $\bR_{3}$ are unitary matrices and $\bD_{3}\eq \diag\set{d_{1},d_{2},d_{3}}$. 
In this case, $\bD_3 = \diag\set{ 0.19394,\, 1.90321, \, 2.70928}$, and thus one of the singular values of $\bG_{3}$ is much smaller than the others (although the product of all singular values yields $1$, as expected).
As will be shown in Section~\ref{sec:entropy_gain_output_disturb}, for a stable $G(z)$ such uneven distribution of singular values arises only when $G(z)$ has non-minimum phase zeros.
The effect of this can be visualized by looking at the image of the cube $[0,1]^{3}$ through $\bG_{3}$ shown in Fig.~\ref{fig:cube}.
\begin{figure}[htbp]
 \centering
 \includegraphics[width = 8 cm, trim= 0 70 0 70, clip]{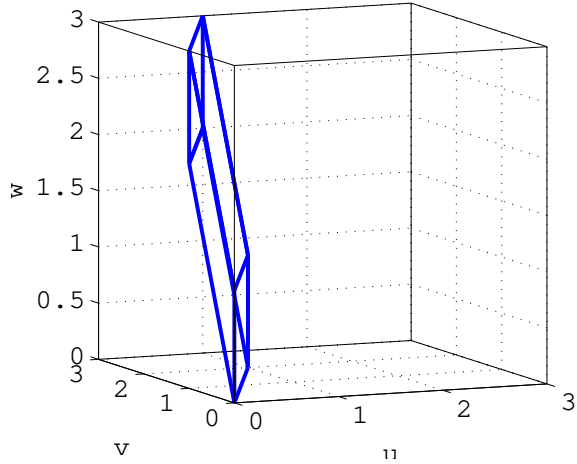} 
\caption{Image of the cube $[0,1]^{3}$ through the square matrix with columns 
$[1\; 2 \; 0]^{T}$,
$[0\; 1 \; 2]^{T}$ and
$[0\; 0 \; 1]^{T}$.
}
\label{fig:cube}
\end{figure}
If the input $\rveu^{1}_{3}$ were uniformly distributed over this cube (of unit volume), then $\bG_{3}\rveu^{1}_{3}$ would distribute uniformly over the unit-volume parallelepiped depicted in Fig.~\ref{fig:cube}, and hence
$h(\bG_{3}\rveu^{1}_{3})=h(\rveu^{1}_{3})$. 

Now, if we add to $\bG_{3}\rveu^{1}_{3}$ a disturbance $\rvez^{1}_{3}=\bPhi \rvas$, with scalar $\rvas$ uniformly distributed over $[-0.5,\ 0.5]$ independent of $\rveu^{1}_{3}$, and with $\bPhi\in\Rl^{3\times 1}$, the effect would be to ``thicken'' the support over which the resulting random vector 
$\rvey^{1}_{3}=\bG_{3}\rveu^{1}_{3}+\rvez^{1}_{3}$ is distributed, along the direction pointed by $\bPhi$.
If $\bPhi$ is aligned with the direction along which the support of $\bG_{3}\rveu^{1}_{3}$ is thinnest
(given by $\bq_{3,1}$, the first row of $\bQ_{3}$), then the resulting support would have its volume significantly increased, which can be associated with a large increase in the differential entropy of $\rvey^{1}_{3}$ with respect to $\rveu^{1}_{3}$.
Indeed, a relatively small variance of $\rvas$ and an approximately aligned $\bPhi$ would still produce a  significant entropy gain.

The above example suggests that the entropy gain from $\rveu^{1}_{n}$ to $\rvey^{1}_{n}$ appears as a combination of two factors.
The first of these is the uneven way in which the random vector $\bG_{n}\rveu^{1}_{n}$ is distributed over $\Rl^{n}$.
The second factor is the alignment of the disturbance vector $\rvez^{1}_{n}$ with respect to the 
span of the subset $\set{\bq_{n,i}}_{i\in\Omega_{n}}$ of columns of $\bQ_{n}$, associated with smallest singular values of $\bG_{n}$, indexed by the elements in the set $\W_n$.
As we shall discuss in the next section, if $G$ has $m$ non-minimum phase zeros, then, as $n$ increases, there will be $m$ singular values of $\bG_{n}$ going to zero exponentially.
Since the product of the singular values of $\bG_{n}$ equals $1$ for all $n$, it follows that $\prod_{i\notin \W_{n}}d_{n,i}$ must grow exponentially with $n$, where $d_{n,i}$ is the $i$-th diagonal entry of $\bD_n$.
This implies that $\bG_{n}\rveu^{1}_{n}$ expands with $n$ along the span of $\set{\bq_{n,i}}_{i\notin\W_{n}}$, compensating its shrinkage along the span of $\set{\bq_{n,i}}_{i\in\W_{n}}$, thus keeping $h(\bG_{n}\rveu^{1}_{n})=h(\rveu^{1}_{n})$ for all $n$.
Thus, as $n$ grows, any small disturbance distributed over the span of $\set{\bq_{n,i}}_{i\in\W_{n}}$, added to $\bG_{n}\rveu^{1}_{n}$, will keep the support of the resulting distribution from shrinking along this subspace.
Consequently, the expansion of 
$\bG_{n}\rveu^{1}_{n}$ with $n$ along the span of $\set{\bq_{n,i}}_{i\notin\W_{n}}$ is no longer compensated, yielding an entropy increase proportional to $\log(\prod_{i\notin\W_{n}} d_{n,i})$.

The above analysis allows one to anticipate a situation in which no entropy gain would take place even when some singular values of $\bG_{n}$ tend to zero as $n\to\infty$. 
Since the increase in entropy is made possible by the fact that, as $n$ grows, the support of the distribution of $\bG_{n}\rveu^{1}_{n}$ shrinks along the span of $\set{\bq_{n,i}}_{i\in\W_{n}}$, no such entropy gain should arise if the support of the distribution of the input $\rveu^{1}_{n}$ expands accordingly along the directions pointed by the rows $\set{\br_{n,i}}_{i\in\W_{n}}$ of $\bR_{n}$. 

An example of such situation can be easily constructed as follows: Let $G(z)$ in Fig.~\ref{fig:general} have non-minimum phase zeros and 
suppose that $\rvau_{1}^{\infty}$ is generated as $G^{-1}\tilde{\rvau}_{1}^{\infty}$, where $\tilde{\rvau}_{1}^{\infty}$ is an i.i.d. random process with bounded entropy rate.
Since the determinant of $\bG_{n}^{-1}$ equals $1$ for all $n$, we have that $h(\rveu^{1}_{n})=h(\tilde{\rveu}^{1}_{n})$, for all $n$.
On the other hand,
$\rvey^{1}_{n}
=
\bG_{n}\bG_{n}^{-1}\tilde{\rveu}^{1}_{n} + \rvez^{1}_{n}
=
\tilde{\rveu}^{1}_{n} + \rvez^{1}_{n}
$.
Since $\rvez^{1}_{n}=[\bPhi]^{1}_{n}\rves^{1}_{\kappa}$ for some finite $\kappa$ (recall Assumption~\ref{assu:z}), it is easy to show that 
$
\lim_{n\to\infty}\frac{1}{n} h(\rvey^{1}_{n})
= 
\lim_{n\to\infty}\frac{1}{n} h(\tilde{\rveu}^{1}_{n})
= 
\lim_{n\to\infty}\frac{1}{n} h({\rveu}^{1}_{n})
$,
and thus no entropy gain appears.

The preceding discussion reveals that the entropy gain produced by $G$ in the situation shown in Fig.~\ref{fig:general} \textbf{depends on the distribution of the input and on the support and distribution of the disturbance}.
This stands in stark contrast with the well known fact that the increase in differential entropy produced by an invertible linear operator depends only on its Jacobian, and not on the statistics of the input~\cite{shawea49}. 
We have also seen that the distribution of a random process along the different directions within the Euclidean space which contains it plays a key role as well.
This motivates the need to specify a class of random processes which distribute more or less evenly over all directions.
The following section introduces a rigorous definition of this class and characterizes a large family of processes belonging to it.

\subsection{Entropy-Balanced Processes}\label{subsec:entropy_balanced}
We begin by formally introducing the notion of an ``entropy-balanced'' process
 $\rvau_{1}^{\infty}$, being one in which, for every finite $\nu\in\Nl$, the differential entropy rate of the orthogonal projection of $\rvau_{1}^{n}$ into any subspace of dimension $n-\nu$ equals the entropy rate of $\rvau_{1}^{n}$ as $n\to\infty$.
This idea is precisely in the following definition.
\begin{defn}\label{def:entropy_balanced}
 A random process $\set{\rvav(k)}_{k=1}^{\infty}$ is said to be entropy balanced if, for every $\nu\in\Nl$, 
\begin{subequations}\label{eq:the_painful_assumption}
\begin{align}
\lim_{n\to\infty}\frac{1}{n}& 
\left( 
h(\boldsymbol\Phi_{n}\rvev^{1}_{n}  )- h(\rvev^{1}_{n})
\right) =0, 
\end{align} 
\end{subequations} 
for every sequence of matrices $\set{\boldsymbol{\Phi}_{n}}_{n=\nu+1}^{\infty}$, $\boldsymbol{\Phi}_n\in\Rl^{(n-\nu)\times n}$, with orthonormal rows. 
\finenunciado
\end{defn}
Equivalently, a random process $\set{\rvav(k)}$ is entropy balanced if every unitary transformation on $\rvav_{1}^{n}$ yields a random sequence $\rvay_{1}^{n}$ such that
$
 \lim_{n\to \infty}\frac{1}{n}|h(\rvay_{n-\nu+1}^{n}|\rvay_{1}^{n-\nu})| =0
$.
This property of the resulting random sequence $\rvay_{1}^{n}$ means that one cannot predict its last $\nu$ samples with arbitrary accuracy by using its previous $n-\nu$ samples, even if $n$ goes to infinity.

We now characterize a large family of entropy-balanced random processes and establish some of their properties.
Although intuition may suggest that most random processes (such as i.i.d. or stationary processes) should be entropy balanced, that statement seems rather difficult to prove.
In the following, we show that the entropy-balanced condition is met by i.i.d. processes with per-sample \textit{probability density function} (PDF) being  uniform, piece-wise constant or Gaussian.
It is also shown that adding to an entropy-balanced process an independent random processes independent of the former yields another entropy-balanced process, and that filtering an entropy-balanced process by a stable and minimum phase filter yields an entropy-balanced process as well.

\begin{prop}\label{prop:gaussian_is_entropy_balanced}
Let $\rvau_{1}^{\infty}$ be a Gaussian i.i.d. random process with positive and bounded per-sample variance.
 Then $\rvau_{1}^{\infty}$ is entropy balanced.\finenunciado
\end{prop}
%
\begin{lem}\label{lem:piecewiseconstant}
 Let $\rvau_{1}^{\infty}$ be an i.i.d. process with finite differential entropy rate, in which each $\rvau_i$ is distributed according to a piece-wise constant PDF in which each interval where this PDF is constant has measure greater than $\epsilon$, for some bounded-away-from-zero constant $\epsilon$. 
 Then $\rvau_{1}^{\infty}$ is entropy balanced.\finenunciado
\end{lem}

\begin{lem}\label{lem:sum_yields_entropy_balanced}
 Let $\rvau_{1}^{\infty}$ and $\rvav_{1}^{\infty}$ be mutually independent random processes.
 If $\rvau_{1}^{\infty}$ is entropy balanced, then $\rvaw_{1}^{\infty}\eq \rvau_{1}^{\infty} + \rvav_{1}^{\infty}$ is also entropy balanced.\finenunciado
\end{lem}
The working behind this lemma can be interpreted intuitively by noting that adding to a random process another independent random process can only increase the ``spread'' of the distribution of the former, which tends to balance the entropy of the resulting process along all dimensions in Euclidean space.
In addition, it follows from Lemma~\ref{lem:sum_yields_entropy_balanced} that all i.i.d. processes having a per-sample PDF which can be constructed by convolving uniform, piece-wise constant or Gaussian PDFs as many times as required are entropy balanced.
It also implies that one can have non-stationary processes which are entropy balanced, since Lemma~\ref{lem:sum_yields_entropy_balanced} imposes no requirements for the process $\rvav_{1}^{\infty}$.

Our last lemma related to the properties of entropy-balanced processes shows that filtering by a stable and minimum phase LTI filter preserves the entropy balanced condition of its input.
\begin{lem}\label{lem:filtering_preserves_entropy_balance}
 Let $\rvau_{1}^{\infty}$ be an entropy-balanced process and $G$ an LTI stable and minimum-phase filter.
 Then the output $\rvaw_{1}^{\infty}\eq G\rvau_{1}^{\infty}$ is also an entropy-balanced process.\finenunciado
\end{lem}
This result implies that any stable moving-average auto-regressive process constructed from entropy-balanced innovations is also entropy balanced, provided the coefficients of the averaging and regression correspond to a stable MP filter.

We finish this section by pointing out two examples of processes which are non-entropy-balanced, namely, 
the output of a NMP-filter to an entropy-balanced input and the output of an unstable filter to an entropy-balanced input. 
The first of these cases plays a central role in the next section.


\section{Entropy Gain due to External Disturbances}\label{sec:entropy_gain_output_disturb}
In this section we formalize the ideas which were qualitatively outlined in the previous section.
Specifically, 
for the system shown in Fig.~\ref{fig:general}
we will characterize the entropy gain $\Gsp(G,\rvex_{0},\rvau_{1}^{\infty},\rvaz_{1}^{\infty})$ defined in~\eqref{eq:Gsp_def} for the case in which the initial state $\rvex_{0}$ is zero (or deterministic) and there exists an  output random disturbance of (possibly infinite length) $\rvaz_{1}^{\infty}$ which satisfies Assumption~\ref{assu:z}.
The following lemmas will be instrumental for that purpose.

\begin{lem}\label{lem:singular_values_bounded}
 Let $A(z)$ be a causal, finite-order, stable and minimum-phase rational transfer function with impulse response $a_{0},a_{1},\ldots$ such that $a_{0}=1$.
 Then 
 $\lim_{n\to\infty}\lambda_{1}(\bA_{n}\bA^{T}_{n})>0$
 and 
 $\lim_{n\to\infty}\lambda_{n}(\bA_{n}\bA^{T}_{n})<\infty$.
 \finenunciado
\end{lem}

\begin{lem}\label{lem:gap_with_two_terms}
Consider the system in Fig.~\ref{fig:general}, and suppose $\rvaz_{1}^{\infty}$ satisfies Assumption~\ref{assu:z}, and that the input process $\rvau_{1}^{\infty}$ is entropy balanced.
 Let $\bG_{n}=\bQ_{n}^{T}\bD_{n}\bR_{n}$ be the SVD of $\bG_{n}$, where $\bD_{n}=\diag\set{d_{n,1},\ldots,d_{n,n}}$ are the singular values of $\bG_{n}$, with $d_{n,1}\leq d_{n,2}\leq\cdots \leq  d_{n,n}$, such that $|\det\bG_{n}| = 1$ $\forall n$. 
Let $m$ be the number of these singular values which tend to zero exponentially as $n\to\infty$.
%
Then
 \begin{align}\label{eq:gap_with_two_terms}
  \lim_{n\to\infty}\frac{1}{n}\left(h(\rvay_{1}^{n}) -h(\rvau_{1}^{n})\right) 
  =
  \lim_{n\to\infty}\frac{1}{n}\left( -\Sumfromto{i=1}{m}   \log d_{n,i} + h\left([\bD_{n}]^{1}_{m}\bR_{n}\rveu^{1}_{n} +[\bQ_{n}]^{1}_{m}\rvez^{1}_{n} \right) \right).
 \end{align}
 \finenunciado
\end{lem}
(The proof of this Lemma can be found in the Appendix, page~\pageref{proof:lem_gap_with_two_terms}).

The previous lemma precisely formulates the geometric idea outlined in Section~\ref{sec:geometric_interpretation}.
To see this, notice that no entropy gain is obtained if the output disturbance vector $\rvez^{1}_{n}$ is orthogonal to the space spanned by the first $m$ columns of $\bQ_{n}$.
If this were the case, then the disturbance would not be able fill the subspace along which $\bG_{n}\rveu^{1}_{n}$ is shrinking exponentially.
Indeed, if $[\bQ_{n}]^{1}_{n}\rvez^{1}_{n}=0$ for all $n$, then
$
h([\bD_{n}]^{1}_{m}\bR_{n}\rveu^{1}_{n} +[\bQ_{n}]^{1}_{m}\rvez^{1}_{n} )
=
h({^{1}}\![\bD_{n}]_{m}[\bR_{n}]^{1}_{m}\rveu^{1}_{n})
=
\sum_{i=1}^{m} \log d_{n,i}+h([\bR_{n}]^{1}_{m}\rveu^{1}_{n})
$,
and the latter sum cancels out the one on the RHS of~\eqref{eq:gap_with_two_terms}, while 
$\lim_{n\to\infty}\frac{1}{n}h([\bR_{n}]^{1}_{n}\rveu^{1}_{n})=0$ since $\rvau_{1}^{\infty}$ is entropy balanced.
On the contrary (and loosely speaking), if the projection of the
support of $\rvez^{1}_{n}$ onto the subspace spanned by the first $m$ rows of $\bQ_{n}$ is of dimension $m$, then
$h([\bD_{n}]^{1}_{m}\bR_{n}\rveu^{1}_{n} +[\bQ_{n}]^{1}_{m}\rvez^{1}_{n} )$ remains bounded for all $n$, and the entropy limit of the sum 
$\lim_{n\to\infty}\frac{1}{n}( -\sumfromto{i=1}{m}   \log d_{n,i})$ on the RHS of~\eqref{eq:gap_with_two_terms} yields the largest possible entropy gain.
Notice that 
$
-\sumfromto{i=1}{m}   \log d_{n,i} 
=
\sumfromto{i=m+1}{n}   \log d_{n,i}
$ (because $\det(\bG_{n})=1$), 
and thus this entropy gain stems from the uncompensated expansion of $\bG_{n}\rveu^{1}_{n}$ along the space spanned by the rows of $[\bQ_{n}]^{m+1}_{n}$.

Lemma~\ref{lem:gap_with_two_terms} also yields the following corollary, which states that 
only a filter $G(z)$ with zeros outside the unit circle (i.e., an NMP transfer function) can introduce entropy gain.
\begin{coro}[Minimum Phase Filters do not Introduce Entropy Gain]\label{coro:MP_filters_no_EG}
Consider the system shown in Fig.~\ref{fig:general} and
let $\rvau_{1}^{\infty}$ be an entropy-balanced random process with bounded entropy rate.
Besides Assumption~\ref{assu:G_Factorized}, suppose that $G(z)$ is minimum phase.
 Then 
 \begin{align}
\lim_{n\to\infty}\frac{1}{n}\left(h(\rvay_{1}^{n}) -h(\rvau_{1}^{n})\right) =0.  
 \end{align}
 \finenunciado
\end{coro}
\begin{proof}
Since $G(z)$ is minimum phase and stable, it follows from Lemma~\ref{lem:singular_values_bounded} that the number of singular values of $\bG_{n}$ which go to zero exponentially, as $n\to\infty$, is zero.
Indeed, all the singular values vary polynomially with $n$.
Thus $m=0$ and Lemma~\ref{lem:gap_with_two_terms} yields directly that the entropy gain is zero (since the RHS of~\eqref{eq:gap_with_two_terms} is zero).
\end{proof}

\subsection{Input Disturbances Do Not Produce Entropy Gain}
In this section we show that random disturbances satisfying Assumption~\ref{assu:z},
when added to the \textit{input} $\rvau_{1}^{\infty}$ (i.e., before $G$), do not introduce entropy gain.
This result can be obtained from Lemma~\ref{lem:gap_with_two_terms}, as stated in the following theorem:
\begin{thm}[Input Disturbances do not Introduce Entropy Gain]
Let $G$ satisfy Assumption~\ref{assu:G_Factorized}. 
Suppose that $\rvau_{1}^{\infty}$ is entropy balanced and consider the output 
 \begin{align}
  \rvay_{1}^{\infty} = G\  ( \rvau_{1}^{\infty} +  \rvab_{1}^{\infty}).
 \end{align}
 where 
$
 \rveb^{1}_{\infty} = \bPsi \rvea^{1}_{\nu},
$
with $\rvea^{1}_{\nu}$ being a random vector satisfying $h(\rvea^{1}_{\nu})<\infty$, and where $\bPsi\in\Rl^{|\Nl|\times \nu}$ has orthonormal columns.
Then,
\begin{align}
 \lim_{n\to\infty}\frac{1}{n}\left(h(\rvay_{1}^{n}) -h(\rvau_{1}^{n})\right) =0
\end{align}
\end{thm}

\begin{proof}
 In this case, the effect of the input disturbance in the output is the forced response of $G$ to it.
 This response can be regarded as an output disturbance $\rvaz_{1}^{\infty} = G \rvab_{1}^{\infty}$. 
 Thus, the argument of the differential entropy on the RHS of~\eqref{eq:gap_with_two_terms} is 
 \begin{align}
  [\bD_{n}]^{1}_{m}\bR_{n}\rveu^{1}_{n} +[\bQ_{n}]^{1}_{m}\rvez^{1}_{n}
  &
  =
  [\bD_{n}]^{1}_{m}\bR_{n}\rveu^{1}_{n} +[\bQ_{n}]^{1}_{m} \bQ_{n}^{T}\bD_{n}\bR_{n}\rveb^{1}_{n}
\\&
=
  [\bD_{n}]^{1}_{m}\bR_{n}\rveu^{1}_{n} +[\bD_{n}]^{1}_{m}\bR_{n}\rveb^{1}_{n}
\\&
=
  \block[\bD_{n}]{1}{m} \rows[\bR_{n}]{1}{m}\left( \rveu^{1}_{n} + \rveb^{1}_{n}\right).
  \end{align}
Therefore,
\begin{align}
  h([\bD_{n}]^{1}_{m}\bR_{n}\rveu^{1}_{n} +[\bQ_{n}]^{1}_{m}\rvez^{1}_{n})
  &
  =
  h(\block[\bD_{n}]{1}{m} \rows[\bR_{n}]{1}{m}\left( \rveu^{1}_{n} + \rveb^{1}_{n}\right))
  \\&
  =
  \sumfromto{i=1}{m}\log d_{n,i} + h(\rows[\bR_{n}]{1}{m}\left( \rveu^{1}_{n} + [\bPsi]^{1}_{n}\rvea^{1}_{\nu}\right) ).
\end{align}
The proof is completed by substituting this result into the RHS of~\eqref{eq:gap_with_two_terms} and noticing that $$\lim_{n\to\infty}\frac{1}{n}h\left([\bR_{n}]^{1}_{m}(\rveu^{1}_{n} +[\bPsi]^{1}_{n}\rvea^{1}_{\nu})\right)=0.$$
\end{proof}

\begin{rem}
 An alternative proof for this result can be given based upon the properties of an entropy-balanced sequence, as follows.
 Since $\det(\bG_{n})=1,\ \forall n$, we have that 
 $
 h(\bG_{n}(\rveu^{1}_{n}+\rveb^{1}_{n}))
= h(\rveu^{1}_{n}+\rveb^{1}_{n})$.
Let $\bTheta_{n}\in\Rl^{\nu\times n}$ and  $\overline{\bTheta}_{n}\in\Rl^{(n-\nu)\times n}$ be a matrices with orthonormal rows which satisfy  
$\overline{\bTheta}_{n}[\bPsi]^{1}_{n}=\bzero$ and such that
 $[\bTheta_{n}^{T} \,|\, \overline{\bTheta}_{n}^{T}]^{T}$ is a unitary matrix.
 Then
\begin{align}
 h([\bTheta_{n}^{T} \,|\, \overline{\bTheta}_{n}^{T}]^{T}\left(\rveu^{1}_{n}+\rveb^{1}_{n}\right))
 =
 h(\bTheta_{n} \rveu^{1}_{n}+ \bTheta_{n}[\bPsi]^{1}_{n}\rvea^{1}_{\nu} \,|\, \overline{\bTheta}_{n}\rveu^{1}_{n})
 +
 h(\overline{\bTheta}_{n}\rveu^{1}_{n}),
\end{align}
where we have applied the chain rule of differential entropy.
But
\begin{align}
 h(\bTheta_{n} \rveu^{1}_{n}+ \bTheta_{n}[\bPsi]^{1}_{n}\rvea^{1}_{\nu} | \overline{\bTheta}_{n}\rveu^{1}_{n})
 \leq  
 h(\bTheta_{n} \rveu^{1}_{n}+ \bTheta_{n}[\bPsi]^{1}_{n}\rvea^{1}_{\nu} )
\end{align}
which is upper bounded for all $n$ because $h(\rvea^{1}_{n})<\infty$ and $h(\bTheta_{n} \rveu^{1}_{n})<\infty$, the latter due to $\rvau_{1}^{\infty}$ being entropy balanced.
On the other hand, since $\rveb^{1}_{n}$ is independent of $\rveu^{1}_{n}$, it follows that $h(\rveu^{1}_{n}+\rveb^{1}_{n})\geq h(\rveu^{1}_{n})$, for all $n$.
Thus 
$
\lim_{n\to\infty}\frac{1}{n}(h(\rvey^{1}_{n}) -h(\rveu^{1}_{n})) 
=
\lim_{n\to\infty}\frac{1}{n}(h(\overline{\bTheta}_{n}\rveu^{1}_{n}) -h(\rveu^{1}_{n}))
=0$,
where the last equality stems from the fact that $\rvau_{1}^{\infty}$ is entropy balanced.
\finenunciado
\end{rem}

\subsection{The Entropy Gain Introduced by Output Disturbances when $G(z)$ has NMP Zeros}
We show here that the entropy gain of a transfer function with zeros outside the unit circle is at most the sum of the logarithm of the magnitude of these zeros.
To be more precise, the following assumption is required. 

\begin{assu}\label{assu:zeros_of_G}
The filter $G$ satisfies Assumption~\ref{assu:G_Factorized} and its transfer function $G(z)$
has $p$ poles and $p$ zeros,  $m$ of which are NMP-zeros.
Let $M$ be the number of distinct NMP zeros, given by
$\{\rho_i\}_{i=1}^M$, i.e., such that $|\rho_1|>|\rho_2|>\dots>|\rho_{M}|>1$, with 
$\ell_i$ being the multiplicity of the $i$-th distinct zero.
We denote by $\iota(i)$, where $\iota:\set{1,\ldots,m}\to\set{1,\ldots,M}$, the distinct zero of $G(z)$ associated with the $i$-th non-distinct zero of $G(z)$, i.e.,
\begin{align}\label{eq:iota_def}
 \iota(k) &\eq 
 \min\set{\iota  :\sumfromto{i=1}{\iota}\ell_i\geq k}.
\end{align}
\finenunciado 
\end{assu}

%

As can be anticipated from the previous results in this section, we will need to characterize the asymptotic behaviour of the singular values of $\bG_{n}$.
This is accomplished in the following lemma, which relates these singular values to the zeros of $G(z)$.
This result is a generalization of the unnumbered lemma in the proof of~{\cite[Theorem~1]{hasari80}} (restated in the appendix as Lemma~\ref{lem:hashimoto}), which holds for FIR transfer functions, to the case of \emph{infinite-impulse response} (IIR) transfer functions (i.e., transfer functions having poles).
\begin{lem}\label{lem:hashimoto_IIR}
For a transfer function $G$ satisfying Assumption~\ref{assu:zeros_of_G}, it holds that
\begin{align}
 \lambda_{l}(\bG_n\bG_n^T) 
 = 
 \begin{cases}  
 \alpha_{n,l}^{2}(\rho_{ \iota(l)})^{-2n}	&, \text{if } l\leq m,\\
 \alpha_{n,l}^{2}				&, \text{otherwise },
 \end{cases}
\end{align}
where the elements in the sequence $\set{\alpha_{n,l} }$ are positive and increase or decrease at most polynomially with $n$. \finenunciado
\end{lem}
(The proof of this lemma can be found in the appendix, page~\pageref{proof:lem_hashimoto_IIR}).

We can now state the first main result of this section.
\begin{thm}\label{thm:eg_n_instate_w_disturb_ineq}
In the system of Fig.~\ref{fig:general}, suppose that $\rvau_{1}^{\infty}$ is entropy balanced and that
$G(z)$ and $\rvaz_{1}^{\infty}$ satisfy assumptions~\ref{assu:zeros_of_G}  and~\ref{assu:z}, respectively.
Then
\begin{align}~\label{eq:eg_n_instate_w_disturb_ineq}
0\leq \lim_{n\to\infty}\frac{1}{n}\left(h(\rvay_{1}^{n}) -h(\rvau_{1}^{n})\right)
 \leq 
 \Sumover{i=1}^{\bar{\kappa}}\log |\rho_{\iota(i)}|,
\end{align}
where
$\bar{\kappa}\eq \min\set{\kappa,m}$ and
$\kappa$ is as defined in Assumption~\ref{assu:z}.
Both  bounds are tight.
The upper bound is achieved if 
$\lim_{n\to\infty}\det(
\rows[\bQ_{n}]{1}{\bar{\kappa}} \rows[\bPhi]{1}{n}(\rows[\bQ_{n}]{1}{\bar{\kappa}} \rows[\bPhi]{1}{n})^{T})>0$,
where the unitary matrices $\bQ_{n}^{T}\in\Rl^{n\times n}$ hold the left singular vectors of $\bG_{n}$.
\finenunciado
\end{thm}
\begin{proof}
See Appendix, page~\pageref{proof:thm_eg_n_instate_w_disturb_ineq}. 
\end{proof}

The second main theorem of this section is the following:
\begin{thm}\label{thm:eg_n_instate_w_disturb}
In the system of Fig.~\ref{fig:general}, suppose that $\rvau_{1}^{\infty}$ is entropy balanced and that
$G(z)$ satisfies Assumption~\ref{assu:zeros_of_G}.
Let $\rvaz_{1}^{\infty}$ be a random output disturbance, such that $\rvaz(i)=0,\, \forall i > m$, and that $|h(\rvaz_{1}^{m})|<\infty$.
Then
\begin{align}
 \lim_{n\to\infty}
 \frac{1}{n}
 \left(h(\rvay_{1}^{n}) - h(\rvau_{1}^{n}) \right) = \Sumfromto{i=1}{m}\log |\rho_{\iota(i)}|.
\end{align}
\finenunciado
\end{thm}

\begin{proof} 
See Appendix, page~\pageref{proof:thm:eg_n_instate_w_disturb}.
\end{proof}

\section{Entropy Gain due to a Random Initial Sate}\label{sec:entropy_gain_initial_stat}

Here we analyze the case in which there exists a random initial state $\rvex_{0}$ independent of the input $\rvau_{1}^{\infty}$, 
and zero (or deterministic) output disturbance.

The effect of a random initial state appears in the output as the natural response of $G$ to it, namely the sequence 
$\bar{\rvay}_{1}^{n}$.
%
Thus, $\rvay_{1}^{n}$ can be written in vector form as
\begin{align}
 \rvey^{1}_{n} = \bG_{n}\rveu^{1}_{n} + \bar{\rvey}^{1}_{n}. 
\end{align}
This reveals that the effect of a random initial state can be treated as a random output disturbance, which allows us to apply the results from the previous sections.

Recall from Assumption~\ref{assu:zeros_of_G} that $G(z)$ is a stable and biproper rational transfer function with $m$ NMP zeros.
As such, it can be factored as 
 \begin{align}\label{eq:G_Factorized_as_tilde_G_F}
  G(z) = P(z)N(z),
  \end{align}
where $P(z)$ is a biproper filter containing only all the poles of $G(z)$, and $N(z)$ is a FIR biproper filter, containing all the zeros of $G(z)$.

We have already established (recall Theorem~\ref{coro:MP_filters_no_EG}) that the entropy gain introduced by the minimum phase system $P(z)$ is zero.
It then follows that the entropy gain can be introduced only by the NMP-zeros of $N(z)$ and an appropriate output disturbance $\bar{\rvay}_1^\infty$.
Notice that, in this case, the input process $\rvaw_1^\infty$ to $N$ (i.e., the output sequence of $P$ due to a random input $\rvau_1^\infty$) is independent of $\bar{\rvay}_1^\infty$ (since we have placed the natural response $\bar{\rvay}_{1}^{\infty}$ after the filters $P$ and $N$, hose initial state is now zero).
This condition allows us to directly use Lemma~\ref{lem:gap_with_two_terms} in order to analyze the entropy gain that $\rvau_1^\infty$ experiences after being filtered by $G$, which coincides with $\bar{h}(\rvay_1^\infty)-\bar{h}(\rvaw_1^\infty)$.
This is achieved by the next theorem.

\begin{thm}\label{thm:eg-due-to-random-xo}
 Consider a stable $p$-th order biproper filter $G(z)$ having $m$ NMP-zeros, and with a random initial state $\rvex_0$, such that $|h(\rvex_0)|<\infty$. 
 Then, the entropy gain due to the existence of a random initial state is 
 \begin{align}
  \lim_{n\to\infty} \frac{1}{n}( h(\rvay_1^n) - h(\rvau_1^n) ) = \sumfromto{i=1}{m} \log\abs{\rho_{\iota(i)}}.
 \end{align}
\end{thm}

\begin{proof}\label{proof:thm_eg-due-to-random-xo}
Being a biproper and stable rational transfer function, $G(z)$ can be factorized as
\begin{align}
 G(z) = P(z)N(z),
\end{align}
where $P(z)$ is a stable biproper transfer function containing only all the poles of $G(z)$ and with all its zeros at the origin,
while $N(z)$ is stable and biproper FIR filter, having all the zeros of $G(z)$.
Let $\tilde{\bC}_n\rvex_0$ and $\bC_{n}\rvex_0$ be the natural responses of the systems $P$ and $N$ to their common random initial state $\rvex_{0}$, respectively, where $\tilde{\bC}_n, \bC_n \in \Rl^{n\times p}$.
Then we can write 
\begin{align}
 \rvey^{1}_{n} 
 = 
 \bG_{n}\rveu^{1}_{n} + \bar{\rvey}^{1}_{n}
 =
  \bN_{n}\underbrace{\bP_{n}\rveu^{1}_{n}}_{\eq \rvew^{1}_{n}} + \bar{\rvey}^{1}_{n}
  =
  \bN_{n}\rvew^{1}_{n} + \bar{\rvey}^{1}_{n}.
\end{align}
Since $P(z)$ is stable and MP, it follows from Corollary~\ref{coro:MP_filters_no_EG} that 
$h(\rvew^{1}_{n})=h(\rveu^{1}_{n})$ for all $n$, and therefore 
\begin{align}
 h(\rvey^{1}_{n}) - h(\rveu^{1}_{n}) 
 = 
 h(\rvey^{1}_{n}) - h(\rvew^{1}_{n}).
\end{align}
Therefore, we only need to consider the entropy gain introduced by the (possibly) non-minimum filter $N$ due to a random output disturbance 
$\rvez^{1}_{n} 
=
\bar{\rvey}^{1}_{n}
= 
\bN_{n}\tilde{\bC}_{n}\rvex_{0}
+
\bC_{n}\rvex_{0}
$, 
which is independent of the input $\rvew^{1}_{n}$.
Thus, the conditions of Lemma~\ref{lem:gap_with_two_terms} are met considering $\bG_{n}=\bN_{n}$,
where now
$\bN_{n} = \bQ_{n}^{T}\bD_{n}\bR_{n}$ is the SVD for $\bN_{n}$, and $d_{n,1}\leq d_{n,2}\leq \cdots \leq d_{n,n}$.
Consequently, it suffices to consider the differential entropy on the RHS of~\eqref{eq:gap_with_two_terms}, whose argument is
\begin{align}
 [\bD_{n}]^{1}_{m}\bR_{n}\rveu^{1}_{n} +[\bQ_{n}]^{1}_{m}\bar{\rvey}^{1}_{n}  
 &
 =
 [\bD_{n}]^{1}_{m}\bR_{n}\rveu^{1}_{n} +[\bQ_{n}]^{1}_{m} 
 \left( \bN_{n}\tilde{\bC}_{n}\rvex_{0} + \bC_{n}\rvex_{0}\right)
 \\&
  =
[\bD_{n}]^{1}_{m}\bR_{n} \left(\rveu^{1}_{n} + \tilde{\bC}_{n}\rvex_{0}\right)
+[\bQ_{n}]^{1}_{m} \bC_{n}\rvex_{0}
 \\&
  =
[\bD_{n}]^{1}_{m}\bR_{n} \rvev^{1}_{n} 
+[\bQ_{n}]^{1}_{m} \bC_{n}\rvex_{0},\label{eq:the_term_of_v_and_x0}
\end{align}
where $\rvev^{1}_{n}\eq \rveu^{1}_{n} + \tilde{\bC}_{n}\rvex_{0}$ has bounded entropy rate and is entropy balanced (since $\tilde{\bC}_{n}\rvex_{0}$ is the natural response of a stable LTI system and because of Lemma~\ref{lem:sum_yields_entropy_balanced}).
We remark that, in~\eqref{eq:the_term_of_v_and_x0}, $\rvev^{1}_{n}$ is not independent of $\rvex_{0}$, which precludes one from using the proof of Theorem~\ref{thm:eg_n_instate_w_disturb_ineq} directly.

On the other hand, since $N(z)$ is FIR of order (at most) $p$, we have that
$\bC_{n}= [ \bE_{p}^T \,|\, \bzero^T\,]^T$, 
where $\bE_p\in\Rl^{p\times p}$ is a non-singular upper-triangular matrix independent of $n$.
Hence, $\bC_{n}\rvex_{0}$ can be written as $[\bPhi]^{1}_{n}\rves^{1}_{p}$, where 
$[\Phi]^{1}_{n} = [\bI_p^T \,|\,\bzero^T]^T$ 
and $\rves^1_p \eq \bE_p\rvex_0$.
According to~\eqref{eq:the_term_of_v_and_x0}, the entropy gain in~\eqref{eq:eg_n_instate_w_disturb_ineq} arises as long as $h([\bQ_n]^1_m\bC_n\rvex_0)$ is lower bounded by a finite constant (or if it decreases sub-linearly as $n$ grows).
Then, we need $[\bQ_n]^1_m[\bPhi]^1_n$ to be a full row-ranked matrix in the limit as $n\to\infty$.
However,
\begin{align}
 \det \left([\bQ_{n}]^{1}_{m}[\bPhi_{n}]^{1}_{n} ([\bQ_{n}]^{1}_{m}[\bPhi_{n}]^{1}_{n})^{T}\right)
 &=
 \det \left([\bQ_{n}^{(p)}]^{1}_{m}([\bQ_{n}^{(p)}]^{1}_{m})^{T}\right),
\end{align}
where $[\bQ_n^{(p)}]^1_m$ denotes the first $p$ columns of the first $m$ rows in $\bQ_n$.
We will now show that these determinants do not go to zero as $n\to\infty$.
Define the matrix $\overline{\bQ}_n\in\Rl^{m\times (p-m)}$ such that  $[\bQ_n^{(p)}]^1_m = [{}^1[\bQ_n]_m \,|\,\, \overline{\bQ}_n]$.
Then, it holds that $\forall\bx\in\Rl^n$,
\begin{align}\label{eq:lowerbound_Qp}
 \norm{([\bQ_n^{(p)}]^1_m)^T \bx}^2 &= \norm{({}^1[\bQ_n]_m)^T\bx}^2 + \norm{(\overline{\bQ}_n)^T \bx}^2 \\
 &\geq \norm{({}^1[\bQ_n]_m)^T\bx}^2 \\
 &\geq \left(\lambda_{\text{min}}({}^1[\bQ_n]_m({}^1[\bQ_n]_m)^T)\right)^{2}.
\end{align}
Hence, the minimum singular value of $[\bQ_n^{(p)}]^1_m$ is lower bounded by the smallest singular value of ${}^1[\bQ_n]_m$, for all $n\geq m$.
But it was shown in the proof of Theorem~\ref{thm:eg_n_instate_w_disturb} (see page~\pageref{proof:thm:eg_n_instate_w_disturb}) that
$\lim_{n\to\infty}\lambda_{\text{min}}({}^1[\bQ_n]_m({}^1[\bQ_n]_m)^T)>0$.
Using this result in~\eqref{eq:lowerbound_Qp} and taking the limit, we arrive to 
\begin{align}
\lim_{n\to\infty}\det \left([\bQ_{n}^{(p)}]^{1}_{m}([\bQ_{n}^{(p)}]^{1}_{m})^{T}\right)
 >
 0.
\end{align}
Thus
\begin{align}
 h\left( [\bD_{n}]^{1}_{m}\bR_{n}\rveu^{1}_{n} +[\bQ_{n}]^{1}_{m}\bar{\rvey}^{1}_{n} \right) 
 &
 =
 h\left([\bD_{n}]^{1}_{m}\bR_{n}\rvev^{1}_{n} + [\bQ_{n}]^{1}_{m} [\bPhi]^{1}_{n}\rves^{1}_{p}\right)
 \end{align}
is upper and lower bounded by a constant independent of $n$  because $\rvav_{1}^{\infty}$ is entropy balanced, $[\bD_{n}]^{1}_{m}$ has decaying entries, and $h(\rvas_1^{p})<\infty$, which means that the entropy rate in the RHS of~\eqref{eq:gap_with_two_terms} decays to zero.
The proof is finished by invoking Lemma~\ref{lem:hashimoto_IIR}.
\end{proof}

Theorem~\ref{thm:eg-due-to-random-xo} allows us to formalize the effect that the presence or absence of a random initial state has on the entropy gain using arguments similar to those utilized in Section~\ref{sec:entropy_gain_output_disturb}.
Indeed, if the random initial state $\rvex_0\in\Rl^p$ has finite differential entropy, then the entropy gain achieves~\eqref{eq:Jensen}, since the alignment between $\rvex_0$ and the first $m$ rows of $\bQ_n$ is guaranteed.
This motivates us to characterize the behavior of the entropy gain (due only to a random initial state), when the initial state $\rvex_0$ can be written as $[\bPhi]^1_p\rves^1_\tau$, with $\tau\leq p$, which means that $\rvex_0$ has an undefined (or $-\infty$) differential entropy.

\begin{coro}\label{coro:eg_due_to_xo_ineq}
 Consider an FIR, $p$-order filter $F(z)$ having $m$ NMP-zeros, such that its random initial state can be written as $\rvex_0 = \bPhi \rves^1_\tau$, 
 where $|h(\rvas_1^\tau)|<\infty$ and $\bPhi\in\Rl^{p\times \tau}$ contains orthonormal rows .
 Then,
  \begin{align}
  \lim_{n\to\infty} ( h(\rvay_1^n) - h(\rvau_1^n) ) \leq \Sumfromto{i=1}{\bar{\tau}} \log\abs{\rho_{\iota(i)}}, \label{eq:eg_due_to_xo_ineq}
 \end{align}
 where $\set{\bar{\tau}} \eq \min\set{m,\tau}$.
 The upper bound in~\eqref{eq:eg_due_to_xo_ineq} is achieved when $[\bQ_n]^1_m\bC_n\bPhi([\bQ_n]^1_m\bC_n\bPhi)^T$ is a non-singular matrix, with $\bC_n$ defined by $\bar{\rvey}^1_n = \bC_n\rvex_0$ (as in Theorem~\ref{thm:eg-due-to-random-xo}).
\end{coro}

\begin{proof}
The effect of the random initial state to the output sequence $\rvay_1^\infty$ can be written as $\rvey^1_n = \bC_n\rvex_0$, where $\bC_n = [\bE_p^T \,|\, \bzero^T]^T \in \Rl^{n\times p}$.
Therefore, if $\bQ^T_n\bD_n\bR_n$ is an SVD for $\bF_n$, it holds that 
\begin{align}
h([\bD_n]^1_n\bR_n\rveu^1_n + [\bQ_n]^1_m \bC_n\bPhi\rves^1_\tau) \label{eq:car}
\end{align}
remains bounded, for $n\to\infty$, if and only if $\lim_{n\to\infty}\det([\bQ_n]^1_m\bC_n\bPhi([\bQ_n]^1_m\bC_n\bPhi)^T)>0$.

Define the rank of $[\bQ_n]^1_m\bC_n\bPhi$ as $\tau_n\in\set{1,\ldots,\bar{\tau}}$.
If $\det([\bQ_n]^1_m\bC_n\bPhi([\bQ_n]^1_m\bC_n\bPhi)^T)=0$, then the lower bound is reached by inserting~\eqref{eq:car} in~\eqref{eq:gap_with_two_terms}.
Otherwise, there exists $L$ large enough such that $\tau_n \geq 1$, $\forall n\geq L$.

We then proceed as the proof of Theorem~\ref{thm:eg_n_instate_w_disturb_ineq}, by considering a unitary $(m\times m)$-matrix $\bH_n$, and a $(\tau_n\times m)$-matrix $\bA_n$ such that
\begin{align}
 \bH_{n}[\bQ_{n}]^{1}_{m} \bC_n\bPhi
 =
 \begin{pmatrix}
  \bA_{n}[\bQ_{n}]^{1}_{m} \bC_n\bPhi
  \\
  \bzero
 \end{pmatrix}, \fspace n\geq L.
\end{align}

This procedure allows us to conclude that $h([\bD_n]^1_n\bR_n\rveu^1_n + [\bQ_n]^1_m \bC_n\bPhi\rves^1_\tau) \leq \sumfromto{i=\tau_n + 1}{m}\log d_{n,i}$, and that the lower limit in the latter sum equals $\bar{\tau}+1$ when $[\bQ_n]^1_m \bC_n\bPhi\rves^1_\tau$ is a full row-rank matrix. 
Replacing the latter into~\eqref{eq:gap_with_two_terms} finishes the proof.
%
\end{proof}

\begin{rem}
 If the random initial state $\rvex_0 = \bPhi\rves^1_\tau$ is generated with $\tau\geq p-m$, then the entropy gain introduced by an FIR minimum phase filter $F$ is at least $\log \rho_1$.
 Otherwise, the entropy gain could be identically zero, as long as the columns of $\bE_n\bPhi(\bE_n\bPhi)^T$ fill only the orthogonal space to the span of the row vectors in $[\bQ_n^{(p)}]^1_m$, where $\bE_n$, $\bPhi$ and $[\bQ_n^{(p)}]^1_m$ are defined as in the proof of Theorem~\ref{thm:eg-due-to-random-xo}.
\end{rem}

Both results, Theorem~\ref{thm:eg-due-to-random-xo} and Corollary~\ref{coro:eg_due_to_xo_ineq}, reveal that the entropy gain arises as long as the effect of the random initial state aligns with the first rows of $\bQ_n$, just as in the results of the previous section.

\section{Effective Entropy Gain due to the Intrinsic Properties of the Filter}\label{sec:effective_entropy}
If there are no disturbances and the initial state is zero, then the first $n$ output samples to an input $\rvau_{1}^{n}$ is given by~\eqref{eq:y_of_u_matrix}.
Therefore, the entropy gain in this case, as defined in~\eqref{eq:Gsp_def}, is zero, regardless of whether or not $G$ is NMP.

Despite the above, there is an interesting question which, to the best of the authors' knowledge, has not been addressed before:
Since in any LTI filter the entire output is longer than the input, what would happen if one compared the differential entropies of the complete output sequence to that of the (shorter) input sequence?
As we show next, a proper definition of this question requires recasting the problem in terms of a new definition of differential entropy.
After providing a geometrical interpretation of this problem, we prove that the (new) entropy gain in this case is exactly~\eqref{eq:Jensen}.

\subsection{Geometrical Interpretation} 
Consider the random vectors 
$\rveu\eq [\rvau_1\ \rvau_2]^{T}$ 
and 
$\rvey\eq [\rvay_1\ \rvay_2 \ \rvay_3]^{T}$  
related via
\begin{align}\label{eq:little_example}
 \begin{bmatrix}
  \rvay_1\\
  \rvay_2\\
  \rvay_3
 \end{bmatrix}
=
\underbrace{\begin{pmatrix}
 1 & 0\\
 2 & 1\\
 0 & 2
\end{pmatrix}}
_{\eq\breve{\bG}_{2}}
 \begin{bmatrix}
  \rvau_1\\
  \rvau_2
 \end{bmatrix}.
\end{align}
Suppose $\rveu$ is uniformly distributed over $[0 ,1]\times[0,1]$.
Applying the conventional definition of differential entropy of a random sequence, we would have that
\begin{align}
 h(\rvay_1,\rvay_2,\rvay_3)
 &
 =
 h(\rvay_1,\rvay_2) + h(\rvay_3|\rvay_1,\rvay_2)
 =
 -\infty,
\end{align}
because $\rvay_3$ is a deterministic function of $\rvay_1$ and $\rvay_2$:
$$
\rvay_3 = [0\;\; 2][\rvau_1 \; \rvau_2]^{T}=
[0\;\; 2]
\begin{pmatrix}
 1 & 0\\
 2 & 1\\
\end{pmatrix}^{-1} 
 \begin{bmatrix}
  \rvay_1\\
  \rvay_2
 \end{bmatrix}.
$$
In other words, the problem lies in that although the output is a  three dimensional vector, it only has two degrees of freedom, i.e., it is restricted to a 2-dimensional subspace of $\Rl^{3}$.
This is illustrated in Fig.~\ref{fig:square_shear}, where the set $[0,1]\times[0,1]$ is shown (coinciding with the \texttt{u}-\texttt{v} plane), together with its image through $\breve{\bG}_{2}$ (as defined in~\eqref{eq:little_example}).
\begin{figure}[htbp]
 \centering
 \includegraphics[width = 8 cm]{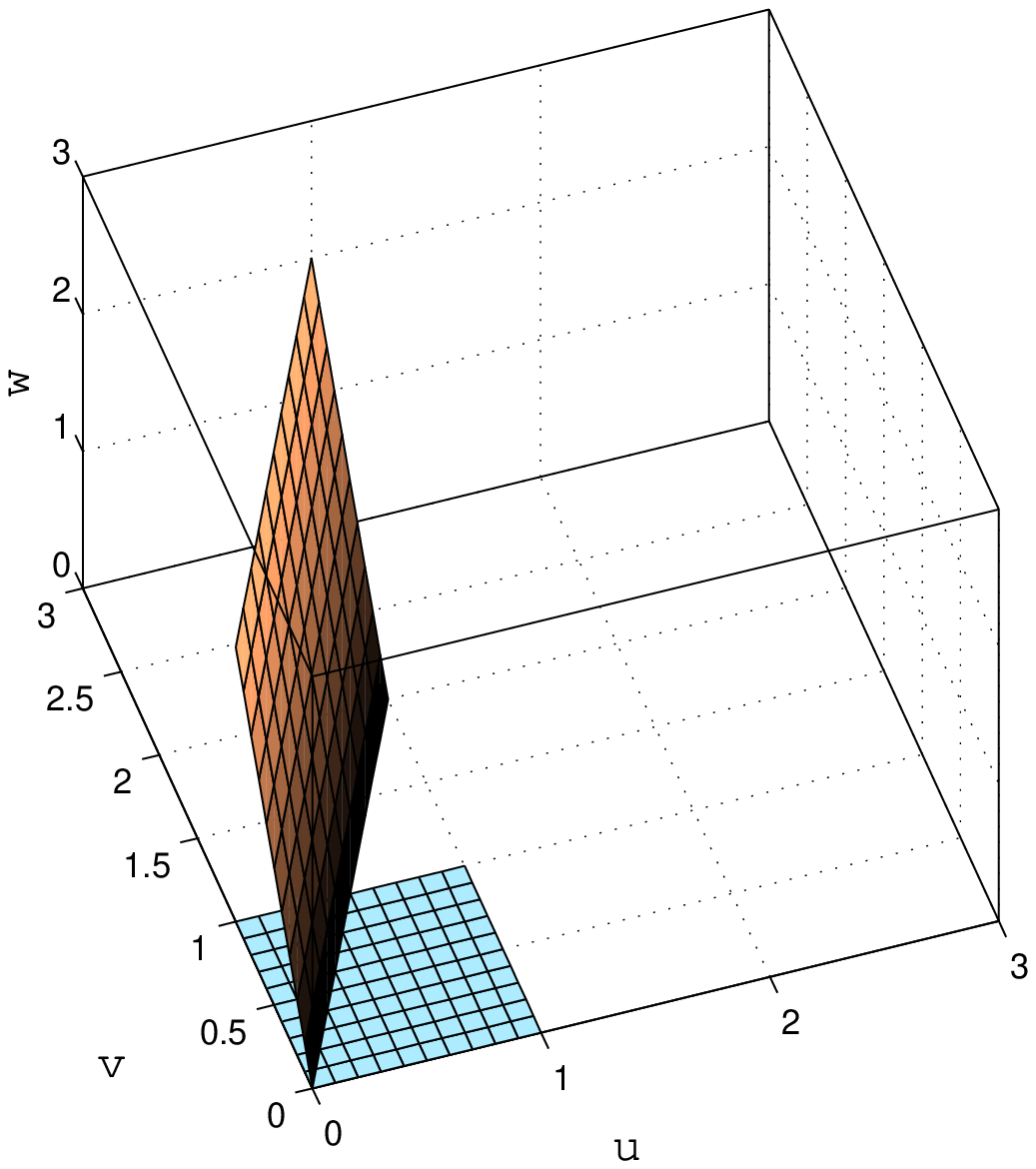}
 \caption{Support of $\rveu$ (laying in the \texttt{u}-\texttt{v} plane) compared to that of $\rvey=\breve{\bG}\rveu$ (the rhombus in $\Rl^3$).}\label{fig:square_shear}
\end{figure}

As can be seen in this figure, the image of the square $[0,1]^{2}$ through $\breve{\bG}_{2}$ is a $2$-dimensional rhombus over which 
$\set{\rvay_1,\rvay_2,\rvay_3}$ distributes uniformly.
Since the intuitive notion of differential entropy of an ensemble of random variables (such as how difficult it is to compress it in a lossy fashion) relates to the size of the region spanned by the associated random vector, one could argue that the differential entropy of $\set{\rvay_1,\rvay_2,\rvay_3}$, far from being $-\infty$, should be somewhat larger than that of $\set{\rvau_1,\rvau_2}$ (since the rhombus $\breve{\bG}_{2}[0,1]^{2}$ has a larger area than $[0,1]^{2}$). 
So, what does it mean that (and why should) $h(\rvay_1,\rvay_2,\rvay_3)=-\infty$?
Simply put, the differential entropy relates to the volume spanned by the support of the probability density function.
For $\rvey$ in our example, the latter (three-dimensional) volume is clearly zero.

From the above discussion, the comparison between the differential entropies of $\rvey\in\Rl^{3}$ and $\rveu\in\Rl^{2}$ of our previous example should take into account that $\rvey$ actually lives in a two-dimensional subspace of $\Rl^{3}$.
Indeed, since the multiplication by a unitary matrix does not alter differential entropies, we could consider the differential entropy of 
\begin{align}
\begin{bmatrix}
\tilde{\rvey}
\\
0
\end{bmatrix}
 \eq 
 \begin{pmatrix}
 \breve{\bQ} 
 \\
 \bar{\bq}^{T}
 \end{pmatrix}
  \rvey,  
\end{align}
where $\breve{\bQ}^T$ is the $3\times2$ matrix with orthonormal rows in the singular-value decomposition of $\breve{\bG}_{2}$
\begin{align}
 \breve{\bG}_{2} = 
  \breve{\bQ}^T \breve{\bD}\, \breve{\bR}.
\end{align}
and $\bar{\bq}$ is a unit-norm vector orthogonal to the rows of $\breve{\bQ}$ (and thus orthogonal to $\rvey$ as well).
We are now able to compute the differential entropy in $\Rl^2$ for $\tilde{\rvey}$, corresponding to the rotated version of $\rvey$ such that its support is now aligned with $\Rl^2$.

The preceding discussion motivates the use of a modified version of the notion of differential entropy for a random vector $\rvey\in\Rl^{n}$ which considers the number of dimensions actually spanned by $\rvey$ instead of its length.

\begin{defn}[The Effective Differential Entropy]\label{def:BMD_entropy}
Let $\rvey\in\Rl^\ell$ be a random vector. 
If $\rvey$ can be written as a linear transformation $\rvey = \bS \rveu$, for some $\rveu\in\Rl^n$ ($ n\leq \ell$), $\bS \in \Rl^{\ell\times n}$, then the effective differential entropy of $\rvey$ is defined as
\begin{align}
 \breve{h}(\rvey) \eq h(\bA\rvey),
\end{align}
where $\bS = \bA^{T}\bT\bC$ is an SVD for $\bS$, with $\bT\in\Rl^{n\times n}$.
\finenunciado
\end{defn}
It is worth mentioning that Shannon's differential entropy
of a vector $\rvey\in\Rl^{\ell}$, whose support's $\ell$-volume is greater than zero, arises from considering it as the difference between its (absolute) entropy and that of a random variable uniformly distributed over an $\ell$-dimensional, unit-volume region of $\Rl^{\ell}$.
More precisely, if in this case the \textit{probability density function} (PDF) of $\rvey = [\rvay_{1}\;\rvay_{2}\; \cdots \; \rvay_{\ell}]^{T}$ is Riemann integrable, then~\cite[Thm.~9.3.1]{covtho91},
\begin{align}\label{eq:h_as_limit_covtho06}
 h(\rvey) = \lim_{\Delta\to 0} \left[H(\rvey^{\Delta}) + \ell\log\Delta\right], 
\end{align}
where $\rvey^{\Delta}$ is the discrete-valued random vector resulting when $\rvey$ is quantized using an $\ell$-dimensional uniform quantizer with $\ell$-cubic quantization cells with volume $\Delta^{\ell}$.
However, if we consider a variable $\rvey$ whose support belongs to an $n$-dimensional subspace of 
$\Rl^\ell$, $n < \ell$ (i.e., $\rvey=\bS\rveu=\bA^{T}\bT\bC\rveu$, as in Definition~\ref{def:BMD_entropy}), 
then the entropy of its quantized version in $\Rl^\ell$, say 
$H_\ell(\rvey^{\Delta})$, is distinct from 
$H_n((\bA\rvey)^{\Delta})$,
the entropy of 
$\bA\rvey$ in 
$\Rl^n$.
Moreover, it turns out that, in general, 
\begin{align}\label{eq:unequal}
 \lim_{\Delta\to 0}\left(H_\ell(\rvey^{\Delta}) -H_n( (\bA\rvey)^{\Delta}) \right) \neq 0,
\end{align}
despite the fact that $\bA$ has orthonormal rows.
Thus, the definition given by~\eqref{eq:h_as_limit_covtho06} does not yield consistent results for the case wherein a random vector has a support's dimension (i.e., its number of degrees of freedom) smaller that its length\footnote{The mentioned inconsistency refers to~\eqref{eq:unequal}, which reveals that the asymptotic behavior $H_\ell(\rvey^{\Delta})$ changes if $\rvey$ is rotated.} 
(If this were not the case, then we could redefine~\eqref{eq:h_as_limit_covtho06} replacing $\ell$ by $n$, in a spirit similar to the one behind Renyi's $d$-dimensional entropy~\cite{renyi-59}.)
To see this, consider the case in which $\rveu\in\Rl$ distributes uniformly over $[0,1]$ and 
$\rvey = [1 \quad 1]^{T}\rveu/\hsqrt{2}$.
Clearly, $\rvey$ distributes uniformly over the unit-length segment connecting the origin with the point $(1,1)/\hsqrt{2}$.
Then 
\begin{align}
 H_{2}(\rvey^{\Delta}) 
 = 
 -
 \left\lfloor \tfrac{1}{\Delta \hsqrt{2}}  \right\rfloor 
  \Delta\hsqrt{2}
  \log \left( \Delta\hsqrt{2}   \right)
  - 
  \left(1- \left\lfloor \tfrac{1}{\Delta \hsqrt{2}}  \right \rfloor \hsqrt{2}\Delta \right)
  \log   \left(1- \left\lfloor \tfrac{1}{\Delta \hsqrt{2}}  \right \rfloor \hsqrt{2}\Delta \right).
\end{align}
On the other hand, since in this case $\bA\rvey=\rveu$, we have that
\begin{align}
 H_{1}((\bA\rvey)^{\Delta})
 =
 H_{1}(\rveu^{\Delta})
 =
 -
 \left\lfloor \tfrac{1}{\Delta }  \right\rfloor 
 \Delta \log\Delta 
 -
 (1-\left\lfloor \tfrac{1}{\Delta }  \right\rfloor \Delta)
 \log (1-\left\lfloor \tfrac{1}{\Delta }  \right\rfloor \Delta).
\end{align}
Thus 
\begin{align}
 \lim_{\Delta\to 0}
 &
 \left(
 H_{1}((\bA\rvey)^{\Delta})
 -
 H_{2}(\rvey^{\Delta}) 
  \right)
  =
  \lim_{\Delta \to 0}
  \left( 
  \left\lfloor \tfrac{1}{\Delta \hsqrt{2}}  \right\rfloor 
  \Delta\hsqrt{2}
  \log \left( \Delta\hsqrt{2}   \right)
  -
 \left\lfloor \tfrac{1}{\Delta }  \right\rfloor 
 \Delta \log\Delta 
  \right)
  =\log\hsqrt{2}.
\end{align}

The latter example further illustrates why the notion of effective entropy is appropriate in the setup considered in this section, where
the effective dimension of the random sequences does not coincide with their length 
(it is easy to verify that the effective entropy of $\rvey$ does not change if one rotates $\rvey$ in $\Rl^{\ell}$).
Indeed, we will need to consider only sequences which can be constructed by multiplying some random vector $\rveu\in\Rl^{n}$, with bounded differential entropy, by a tall matrix $\breve{\bG}_n\in\Rl^{n\times (n+\eta)}$, with $\eta>0$ (as in~\eqref{eq:little_example}), which are precisely the conditions required by Definition~\ref{def:BMD_entropy}. 

\subsection{Effective Entropy Gain}
We can now state the main result of this section:
\begin{thm}\label{thm:BMD_entropy_gain}
 Let the entropy-balanced random sequence $\rvau_{1}^{\infty}$ be the input of an LTI filter $G$, and let $\rvay_{1}^{\infty}$ be its output.
 Assume that $G(z)$ is the $z$-transform of the $(\eta+1)$-length sequence $\set{g_k}_{k=0}^{\eta}$.
 Then 
 \begin{align}
  \lim_{n\to \infty} \frac{1}{n}\left(\breve{h}(\rvay_{1}^{n+\eta}) -\breve{h}(\rvau_{1}^{n}) \right)
  =\intpipi{\log\abs{G\ejw} }.
 \end{align}
 \finenunciado
\end{thm}

Theorem~\ref{thm:BMD_entropy_gain} states that, when considering the full-length output of a filter, the effective entropy gain is introduced by the filter itself, without  requiring the presence of external random disturbances or initial states.
This may seem a surprising result, in view of the findings made in the previous sections, where the entropy gain appeared only when such random exogenous signals were present. 
In other words, when observing the full-length output and the input, the (maximum) entropy gain of a filter can be recasted in terms of the ``volume'' expansion yielded by the filter as a linear operator, provided we measure effective differential entropies instead of Shannon's differential entropy.

\begin{proof}[Proof of Theorem~\ref{thm:BMD_entropy_gain}]
The total length of the output $\ell$, will grow with the length $n$ of the input, if $G$ is FIR, and will be infinite, if $G$ is IIR.
Thus, we define the \textit{output-length function}
\begin{align}
 \ell(n) \eq \text{length of $\rvey$ when input is $\rveu^{1}_{n}$}
 =
 \begin{cases}
  n+\eta & , \text{ if $G$ is FIR with i.r. length $\eta+1$,}\\
  \infty& , \text{ if $G$ is IIR.}
 \end{cases}
\end{align}
It is also convenient to define the sequence of matrices $\set{\breve{\bG}_{n}}_{n=1}^{\infty}$, where
$\breve{\bG}_{n}\in\Rl^{\ell(n)\times n}$ is Toeplitz with 
$\left[\breve{\bG}_{n}\right]_{i,j}=0,\forall i<j$,
$\left[\breve{\bG}_{n}\right]_{i,j}=g_{i-j},\forall i\geq j$.
This allows one to write
the \emph{entire} output $\rvay_1^{\ell}$ of a causal LTI filter $G$ with impulse response $\set{g_{k}}_{k=0}^{\eta}$ to an input $\rvau_1^\infty$ as 
\begin{align}\label{eq:yFu_tall}
\rvey^{1}_{\ell(n)} 
= 
\breve{\bG}_{n}
\rveu^{1}_{n}.
\end{align}
%
Let the SVD of $\breve{\bG}_{n}$ be
$\breve{\bG}_{n} = \breve{\bQ}^T_n\breve{\bD}_n\breve{\bR}_n$, 
where $\breve{\bQ}_n\in\Rl^{n\times \ell(n)}$ has orthonormal rows,
$\breve{\bD}_n\in\Rl^{n\times n}$ is diagonal with positive elements, 
and 
$\breve{\bR}_n\in\Rl^{n\times n}$ is unitary.

The effective differential entropy of $\rvay_1^{n(\ell)}$ exceeds the one of $\rvau_1^{n}$ by 
\begin{align}\label{eq:hybreve_of_logdetD}
\breve{h}(\rvey^1_{n(\ell)}) - \breve{h}(\rveu^1_{n}) &= h(\breve{\bQ}_n \breve{\bG}_n\rveu^1_{n}) - h(\rveu^1_{n}) \\
&= h(\breve{\bD}_n\breve{\bR}_n\rveu^1_n) - h(\rveu^1_n) \\
&= \log \det \breve{\bD}_n,
\end{align}
where the first equality follows from the fact that $\rveu^1_n$ can be written as $\bI_n\rveu^1_n$, which means that 
$\breve{h}(\rveu^1_n) = h(\rveu^1_n)$.
But
\begin{align}
 \breve{\bG}_{n}^{T}\breve{\bG}_{n} 
 = (\breve{\bQ}^T_n\breve{\bD}_n\breve{\bR}_n)^{T}(\breve{\bQ}^T_n\breve{\bD}_n\breve{\bR}_n)
 = \breve{\bR}_n^T\breve{\bD}_n\breve{\bQ}_n\breve{\bQ}_n^T\breve{\bD}_n\breve{\bR}_n
 = \breve{\bR}_n^T\breve{\bD}_n^2\breve{\bR}_n. \label{eq:casicasi}
\end{align}
Since $\breve{\bR}_n$ is unitary, it follows that 
$\det\breve{\bD}_n^{2}= \det\breve{\bG}_{n}^{T}\breve{\bG}_{n}$
, which means that  
$\det\breve{\bD}_n = \frac{1}{2}\det\breve{\bG}_{n}^{T}\breve{\bG}_{n}$.
The product 
$\bH_{n} \eq  \breve{\bG}_{n}^{T} \breve{\bG}_{n}$
%
%
is a symmetric Toeplitz matrix, with its first column,  $[h_{0}\, h_{1}\, \cdots \,h_{n-1}]^{T}$,  given by
\begin{align}
 h_{i} & =  \Sumfromto{k=0}{n}g_{k}g_{k-i}.
\end{align}
Thus, the sequence $\set{h_i}_{i=0}^{n-1}$ corresponds to the samples $0$ to $n-1$ of those resulting from the complete convolution $g*g^{-}$, even when the filter $G$ is IIR, where $g^{-}$ denotes the time-reversed (perhaps infinitely large) response $g$.
Consequently, using the Grenander and Szeg{\"o}'s theorem~\cite{gresze58}, it holds that
\begin{align}\label{eq:GSz}
\lim_{n\to\infty} \log\left( \det(\breve{\bG}^T_{n}\breve{\bG}_n)^{1/n}\right) = \Intpipi{\log \abs{G\ejw} },
\end{align}
where $G\ejw$ is the discrete-time Fourier transform of $\set{g_{k}}_{k=0}^{\ell}$.

In order to finish the proof, we divide~\eqref{eq:hybreve_of_logdetD} by $n$, take the limit as $n\to\infty$, and replace~\eqref{eq:GSz} in the latter.
\end{proof}

\section{Some Implications}\label{sec:implications}
\subsection{Rate Distortion Function for Non-Stationary Processes}
In this section we obtain a simpler proof of a result by 
Gray, 
Hashimoto and Arimoto~\cite{gray--70,hasari80,grahas08},
which compares the rate distortion function (RDF) of 
a non-stationary auto-regressive Gaussian process $\rvax_{1}^{\infty}$ (of a certain class) to that of a corresponding stationary version, under MSE distortion.
Our proof is based upon the ideas developed in the previous sections, and extends the class of non-stationary sources for which the results in~\cite{gray--70,hasari80,grahas08} are valid.

To be more precise,
let
$\set{a_i}_{i=1}^{\infty}$ 
and 
$\set{\tilde{a}_i}_{i=1}^{\infty}$ 
be the impulse responses of two linear time-invariant filters $A$ and $\tilde{A}$ with rational transfer functions
\begin{align}
 A(z) &= \frac{z^{M}}{\prod_{i=1}^{M}(z-p_{i})}\\
 \tilde{A}(z) &= \frac{z^{M}}{\prod_{i=1}^{M}|p_{i}^{*}|(z-1/p_{i}^{*})},
\end{align}
where $\abs{p_{i}}>1$, $\forall i=1,\ldots, M$.
From these definitions it is clear that $A(z)$ is unstable, $\tilde{A}(z)$ is stable, and
\begin{align}
 |A\ejw| = |\tilde{A}\ejw|, \fspace \forallwinpipi.
\end{align}
Notice also that 
$\lim_{\abs{z}\to\infty}A(z)=1$
and
$\lim_{\abs{z}\to\infty}\tilde{A}(z)=1/\prod_{i=1}^{M}|p_{i}|$, and thus 
\begin{align}
 a_{0}=1, && \tilde{a}_{0}= \prod_{i=1}^{M}|p_{i}|^{-1}.
\end{align}

Consider the non-stationary random sequences (source) 
$\rvax_{1}^{\infty}$ and
the asymptotically stationary source
$\tilde{\rvax}_{1}^{\infty}$
generated by passing  a stationary Gaussian process $\rvaw_{1}^{\infty}$ through $A(z)$ and $\tilde{A}(z)$, respectively, which can be written as
\begin{align}
 \rvex^{1}_{n} &= \bA_{n}\rvew^{n}_{1}, \fspace n=1,\ldots,\label{eq:x_def}\\
 \tilde{\rvex}^{1}_{n} &= \tilde{\bA}_{n}\rvew^{n}_{1}, \fspace n=1,\ldots.\label{eq:tildex_def}
\end{align}
(A block-diagram associated with the construction of $\rvax$ is presented in Fig.~\ref{fig:rdfns}.)
\begin{figure}[t]
\centering
\input{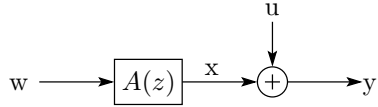}
\caption{Block diagram representation of how the non-stationary source $\rvax_{1}^{\infty}$ is built and then reconstructed as $\rvay=\rvax+\rvau$.}
\label{fig:rdfns}
\end{figure}
Define the rate-distortion functions for these two sources  as 
\begin{align}
R_{\rvax}(D) &\eq  \lim_{n\to\infty}  R_{\rvax,n}(D),
&
R_{\rvax,n}(D)&\eq \min \frac{1}{n}I(\rvax_{1}^{n};\rvax_{1}^{n}+\rvau_{1}^{n}),
\\
R_{\tilde{\rvax}}(D) &\eq  \lim_{n\to\infty}  R_{\tilde{\rvax},n}(D),
&
R_{\tilde{\rvax},n}(D)&\eq \min \frac{1}{n}I(\tilde{\rvax}_{1}^{n};\tilde{\rvax}_{1}^{n}+\tilde{\rvau}_{1}^{n}),
\end{align}
where, for each $n$, the minimums are taken over all the conditional probability density functions 
$f_{\rvau_{1}^{n}|\rvax_{1}^{n}}$ 
and
$f_{\tilde{\rvau}_{1}^{n}|\tilde{\rvax}_{1}^{n}}$ 
yielding 
 $\Expe{\norm{\rveu^{1}_{n}}^{2}}/n\leq D$ and
 $\Expe{\norm{\tilde{\rveu}^{1}_{n}}^{2}}/n\leq D$,
respectively.

The above rate-distortion functions have been characterized in~\cite{gray--70,hasari80,grahas08} for the case in which $\rvaw_{1}^{\infty}$ is an i.i.d. Gaussian process.
In particular, it is explicitly stated in~\cite{hasari80,grahas08} that, for that case, 
\begin{align}
 R_{\rvax}(D)- 
 R_{\tilde{\rvax}}(D)
 =
 \intpipi{\log|A^{-1}\ejw|}
 =
 \sumfromto{i=1}{M}\log|p_{i}|.\label{eq:gap}
\end{align}
We will next provide an alternative and simpler proof of this result, and extend its validity for general (not-necessarily stationary) Gaussian $\rvaw_{1}^{\infty}$, using the entropy gain properties of non-minimum phase filters established in Section~\ref{sec:entropy_gain_output_disturb}.
Indeed, the approach in~\cite{gray--70,hasari80,grahas08} is based upon asymptotically-equivalent Toeplitz matrices in terms of the signals' covariance matrices.
This restricts $\rvaw_{1}^{\infty}$ to be Gaussian and i.i.d. and $A(z)$ to be an all-pole unstable transfer function, and then, the only non-stationary allowed is that arising from unstable poles.
For instance, a cyclo-stationarity innovation followed by an unstable filter $A(z)$ would yield a source which cannot be treated using Gray and Hashimoto's approach.
By contrast, the reasoning behind our proof lets $\rvaw_1^\infty$ be any Gaussian process, and then let the source be $A\rvaw$, with $A(z)$ having unstable poles (and possibly zeros and stable poles as well).

The statement is as follows:
\begin{thm}\label{thm:RDF_non_stat}
 Let $\rvaw_{1}^{\infty}$ be any Gaussian stationary process with bounded differential entropy rate, and let 
 $\rvax_{1}^{\infty}$ and $\tilde{\rvax}_{1}^{\infty}$ be as defined in~\eqref{eq:x_def} and~\eqref{eq:tildex_def}, respectively.
 Then~\eqref{eq:gap} holds.
 \finenunciado
\end{thm}

Thanks to the ideas developed in the previous sections, it is possible to give an intuitive outline of the proof of this theorem (given in the appendix, page~\pageref{proof:RDF_non_stat}) by using a sequence of block diagrams.
More precisely, consider the diagrams shown in Fig.~\ref{fig:rdfnsproof}.
\begin{figure}[t]
\centering
\input{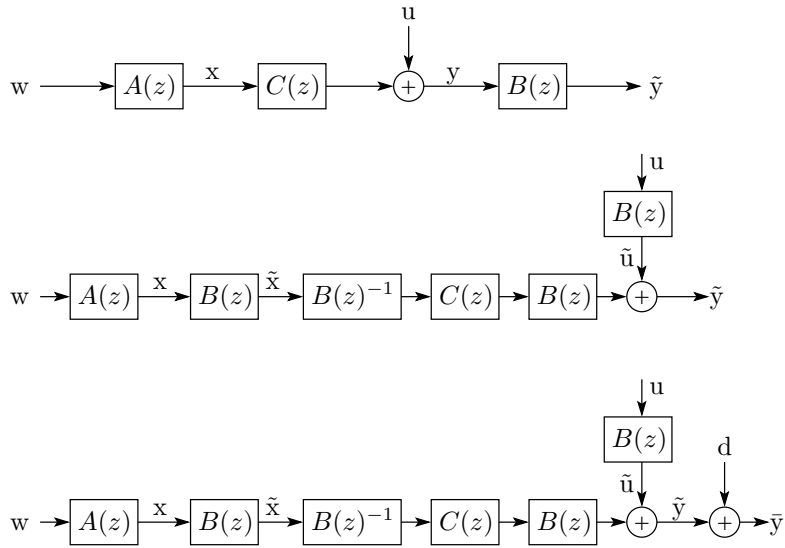}
\caption{Block-diagram representation of the changes of variables in the proof of Theorem~\ref{thm:RDF_non_stat}.}
\label{fig:rdfnsproof}
\end{figure}
In the top diagram in this figure, suppose that $\rvay=C\rvax+\rvau$ realizes the RDF for the non-stationary 
source $\rvax$.
The sequence $\rvau$ is independent of $\rvex$, and the linear filter $C(z)$ is such that the error $(\rvay-\rvax)\Perp \rvay$ 
(a necessary condition for minimum MSE optimality).
The filter $B(z)$ is the Blaschke product of $A(z)$ (see~\eqref{eq:B_def} in the appendix) (a stable, NMP filter with unit frequency response magnitude such that 
$\tilde{\rvax} = B \rvax$).

If one now moves the filter $B(z)$ towards the source, then the middle diagram in Fig.~\ref{fig:rdfnsproof} is obtained.
By doing this, the stationary source $\tilde{\rvax}$ appears with an additive error signal 
$\tilde{\rvau}$ that has the same asymptotic variance as $\rvau$, reconstructed as 
$\tilde{\rvay}=C\tilde{\rvax} + \tilde{\rvau}$.
From the invertibility of $B(z)$, it also follows that the mutual information rate between 
$\tilde{\rvax}$ and $\tilde{\rvay}$
equals that between 
$\rvax$ and $\rvay$.
Thus, the channel 
$\tilde{\rvay}=C\tilde{\rvax}+\tilde{\rvau}$ 
has the same rate and distortion as the channel
$\rvay =C\rvax+\rvau$.

However, if one now adds a short disturbance $\rvad$ to the error signal $\tilde{\rvau}$ 
(as depicted in the bottom diagram of Fig.~\ref{fig:rdfnsproof}),
then the resulting additive error term $\bar{\rvau}=\tilde{\rvau}+\rvad$ will be independent of $\tilde{\rvax}$ and
will have the same asymptotic variance as $\tilde{\rvau}$.
However, the differential entropy rate of $\bar{\rvau}$ will exceed that of $\tilde{\rvau}$ by the RHS of~\eqref{eq:gap}.
This will make the mutual information rate between 
$\tilde{\rvax}$ and $\bar{\rvay}$ to be less than that between 
$\tilde{\rvax}$ and $\tilde{\rvay}$ by the same amount.
Hence, 
$R_{\tilde{\rvex}}(D)$ be at most 
$R_{\rvex}(D) - \sumfromto{i=1}{M}\log\abs{p_{i}}$.
A similar reasoning can be followed to prove that 
$R_{\rvex}(D)-R_{\tilde{\rvex}}(D)\leq  \sumfromto{i=1}{M}\log\abs{p_{i}}$.

%

%

\subsection{Networked Control}
Here we revisit the setup shown in Fig.~\ref{fig:fbksystem} and discussed in Section~\ref{sec:Intro}.
Recall  from~\eqref{eq:martins_I_bound} that, for this general class of networked control systems, it was shown in~\cite[Lemma 3.2]{mardah05} that
\begin{align}
 \lim_{n\to\infty}\frac{1}{n}I(\rvex_{0}; \rvay_{1}^{n})
 \geq \Sumover{\abs{p_{i}}>1}\log\abs{p_{i}},
\end{align}
where $\set{p_{i}}_{i=1}^{M}$ are the poles of $P(z)$ (the plant in Fig.~\ref{fig:fbksystem}).

By using the results obtained in Section~\ref{sec:entropy_gain_initial_stat} we show next that equality holds in~\eqref{eq:martins_I_bound} provided the feedback channel 
satisfies the following assumption:

\begin{assu}\label{assu:fbck_channel}
The feedback channel in Fig.~\ref{fig:fbksystem} can be written as
\begin{align}
 \rvaw= AB \rvav + BF(\rvac), \label{eq:channel}
\end{align}
where
\begin{enumerate}
 \item $A$ and $B$ are stable rational transfer functions such that $AB$ is biproper,  $ABP$  has the same unstable poles as $P$, and the feedback $AB$ stabilizes the plant $P$.
 \item $F$ is any (possibly non-linear) operator such that $\tilde{\rvac}\eq F(\rvac)$ satisfies 
 $\frac{1}{n}h(\tilde{\rvac}_{1}^{n})<K$, for all $n\in\Nl$, and
 \item $\rvac_{1}^{\infty}\Perp (\rvex_{0},\rvau_{1}^{\infty})$.
 \finenunciado
 \end{enumerate}
 \end{assu}

 An illustration of the class of feedback channels satisfying this assumption is depicted on top of Fig.~\ref{fig:fbkplant}.
 Trivial examples of channels satisfying Assumption 5 are a Gaussian additive channel preceded and followed by linear operators~\cite{elia-04}.
 Indeed, when $F$ is an LTI system with a strictly causal transfer function, the feedback channel that satisfies Assumption~\ref{assu:fbck_channel} is widely known as a \emph{noise shaper with input pre and post filter}, used in, e.g.~\cite{silder11,silder11b,yuksel12,fremid11}.
 \begin{figure}[t]
  \centering
  \input{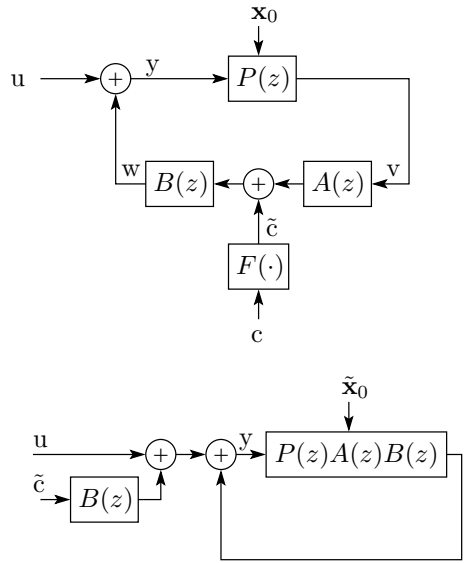}
  \caption{Top: The class of feedback channels described by Assumption~\ref{assu:fbck_channel}. Bottom: an equivalent form.}
  \label{fig:fbkplant}
 \end{figure}

 \begin{thm}\label{thm:equality_in_martins}
  In the networked control system of Fig.~\ref{fig:fbksystem}, suppose that the feedback channel satisfies Assumption~\ref{assu:fbck_channel} and that the input $\rvau_{1}^{\infty}$ is entropy balanced.
  If the random initial state of the plant $P(z)$, with poles $\set{p_{i}}_{i}^{M}$, satisfies $|h(\rvex_{0})|<\infty$, then
  \begin{align}
   \lim_{n\to\infty}\frac{1}{n}I(\rvex_{0}; \rvay_{1}^{n})
   = \sumover{\abs{p_{i}}>1}\log\abs{p_{i}}.
  \end{align}
\finenunciado
 \end{thm}
\begin{proof}
 Let 
 $P(z)=N(z)/\Lambda (z)$ and 
 $T(z)\eq A(z)B(z)=\Gamma(z)/\Theta(z)$.
 Then, from Lemma~\ref{lem:initial_states} (in the appendix), 
the output $\rvey^1_n$ can be written as
 \begin{align}
  \rvay = \underbrace{\Lambda}_{\text{init. state $\rvex_{0}$}} \cdot 
	  \underbrace{\frac{\Theta }{\Theta \Lambda + \Gamma N}}_{\eq \tilde{G}\text{, init. state $\set{\rvex_{0},\rves_{0}}$}} \tilde{\rvau},
	  \label{eq:y_as_Gu}
 \end{align}
where $\rves_{0}$ is the initial state of $T(z)$ and 
\begin{align}
 \tilde{\rvau} \eq u + B \tilde{\rvac}.
\end{align}
(see Fig.~\ref{fig:fbkplant} Bottom).
Then 
\begin{align}
 I(\rvex_{0};\rvey^{1}_{n}) 
 &
 =
 h(\rvey^{1}_{n}) -   h(\rvey^{1}_{n}|\rvex_{0})
 \\&
 =
 h(\rvey^{1}_{n}) -   h (\boldsymbol{\Lambda}_{n} [\tilde{\bG}_{n}\tilde{\rveu}^{1}_{n} + \tilde{\bC}_{n}\rves_{0} ]  )
 \\&
 =
 h(\boldsymbol{\Lambda}_{n} [\tilde{\bG}_{n}\tilde{\rveu}^{1}_{n} + \tilde{\bC}_{n}\rves_{0} +\bar{\bC}_{n}\rvex_{0}] + \bC_{n}\rvex_{0}) -   h(\boldsymbol{\Lambda}_{n} [\tilde{\bG}_{n}\tilde{\rveu}^{1}_{n} + \tilde{\bC}_{n}\rves_{0} ] )
 \\&
 =
 h(\boldsymbol{\Lambda}_{n} [\tilde{\bG}_{n}\tilde{\rveu}^{1}_{n} + \tilde{\bC}_{n}\rves_{0} +\bar{\bC}_{n}\rvex_{0}] + \bC_{n}\rvex_{0}) -   
 h(\tilde{\bG}_{n}\tilde{\rveu}^{1}_{n} + \tilde{\bC}_{n}\rves_{0}  ),
 \label{eq:I_as_h}
\end{align}
 where $\tilde{\bC}_{0}$ maps the initial state $\rves_{0}$ to $\rvey^{1}_{n}$,
 $\bar{\bC}_{n}$ maps the initial state $\rvex_{0}$ to the output of $\tilde{G}(z)$,
 and $\bC_{n}$ maps the initial state $\rvex_{0}$ (of $\Lambda(z)$) to $\rvey^{1}_{n}$.
Since $\rvau_{1}^{\infty}$ is entropy balanced and $\tilde{\rvac}_{1}^{\infty}$ has finite entropy rate, it follows from Lemma~\ref{lem:sum_yields_entropy_balanced} that $\tilde{\rvau}_{1}^{\infty}$ is entropy balanced as well.
Thus, we can proceed as in the proof of Theorem~\ref{thm:eg-due-to-random-xo} to conclude that
 \begin{align}
\lim_{n\to\infty}\frac{1}{n}I(\rvex_{0}; \rvay_{1}^{n})
=
\sumover{\abs{p_{i}}>1}\log\abs{p_{i}}.
\label{eq:I_as_sum_log}
 \end{align}
 This completes the proof.
\end{proof}

%

%

\subsection{The Feedback Channel Capacity of (non-white) Gaussian Channels}
Consider a non-white additive Gaussian channel of the form
\begin{align}
 \rvay_{k}= \rvax_{k} + \rvaz_{k},
\end{align}
where the input $\rvax$ is subject to the power constraint
\begin{align}
 \lim_{n\to\infty}\frac{1}{n}\expe{\norm{\rvex^{1}_{n}}^{2}}\leq P,
\end{align}
and
$\rvaz_{1}^{\infty}$ is a stationary Gaussian process. 

The feedback information capacity of this channel is realized by a Gaussian input $\rvax$, and is given by 
\begin{align}
 C_{\text{FB}} = \lim_{n\to\infty} 
 \max_{\bK_{\rvex^{1}_{n}}: \frac{1}{n}\tr{\bK_{\rvex^{1}_{n}}}\leq P}
 I(\rvex^{1}_{n};\rvey^{1}_{n}),
\end{align}
where $\bK_{\rvex^{1}_{n}}$ is the covariance matrix of $\rvex^{1}_{n}$ and, for every $k\in\Nl$, the input $\rvex_{k}$ is allowed to depend upon the channel outputs $\rvay_{1}^{k-1}$ (since there exists a causal, noise-less feedback channel with one-step delay).

In~\cite{kim-yh10}, it was shown that 
if $\rvaz$ is an auto-regressive moving-average process of $M$-th order, then
$C_{\text{FB}}$ can be achieved by the scheme shown in Fig.~\ref{fig:Kim_FBK_Cap_system}.
In this system, $B$ is a strictly causal and stable finite-order filter and $\rvav_{1}^{\infty}$ is Gaussian with $\rvav_{k}=0$ for all $k>M$ and such that
$\rvev^{1}_{n}$ is Gaussian with a positive-definite covariance matrix $\bK_{\rvev^{1}_{M}}$.
\begin{figure}[t]
\centering
 \input{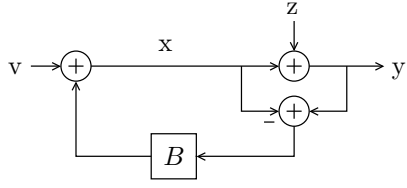}
 \caption{Block diagram representation a non-white Gaussian channel $\rvay=\rvax+\rvaz$ and the coding scheme considered in~\cite{kim-yh10}.}
 \label{fig:Kim_FBK_Cap_system}
\end{figure}

Here we use the ideas developed in Section~\ref{sec:entropy_gain_output_disturb} to show that \textbf{the information rate achieved by the capacity-achieving scheme proposed in~\cite{kim-yh10} drops to zero if there exists any additive disturbance of length at least $M$ and finite differential entropy affecting the output, no matter how small}.

To see this, notice that, in this case, and for all $n>M$, 
\begin{align}
 I(\rvax_{1}^{n};\rvay_{1}^{n}) 
 &= I(\rvav_{1}^{M};\rvay_{1}^{n})
 =
 h(\rvey^{1}_{n}) 
 - 
 h(\rvey^{1}_{n}|\rvev^{1}_{n}) 
 \\&
 =
 h(\rvey^{1}_{n}) 
 - 
 h(( \bI_{n}+\bB_{n})\rvez^{1}_{n} + \rvev^{1}_{n}|\rvev^{1}_{M}) 
 \\&
 =
 h(\rvey^{1}_{n}) 
 - 
 h(( \bI_{n}+\bB_{n})\rvez^{1}_{n} |\rvev^{1}_{M})
 \\&
 =
 h(\rvey^{1}_{n}) 
 - 
 h(( \bI_{n}+\bB_{n})\rvez^{1}_{n}) 
 =
 h(\rvey^{1}_{n}) 
 - 
 h(\rvez^{1}_{n})
 \\&
 =
 h(( \bI_{n}+\bB_{n})\rvez^{1}_{n} + \rvev^{1}_{n})
 -
 h(\rvez^{1}_{n}),
\end{align}
since $\det(\bI_{n}+\bB_{n})=1$.
From Theorem~\ref{thm:eg_n_instate_w_disturb}, this gap between differential entropies is precisely the entropy gain introduced by $\bI_{n}+\bB_{n}$ to an input $\rvez^{1}_{n}$ when the output is affected by the disturbance $\rvev^{1}_{M}$.
Thus, from Theorem~\ref{thm:eg_n_instate_w_disturb}, the capacity of this scheme will correspond to
$\intpipi{\log\abs{1+B\ejw}}=\sumover{\abs{\rho_{i}}>1}\log\abs{\rho_{i}}$, where $\set{\rho_{i}}_{i=1}^{M}$ are the zeros of  $1+B(z)$, which is precisely the result stated in~\cite[Theorem 4.1]{kim-yh10}.

However, if the output is now affected by an additive disturbance $\rvad_{1}^{\infty}$ not passing through $B(z)$ such that $\rvad_{k}=0$, $\forall k>M$ and $|h(\rved^{1}_{M})|<\infty$, with $\rvad_{1}^{\infty}\Perp (\rvav_{1}^{M},\rvaz_{1}^{\infty})$, then we will have 
\begin{align}
 \rvey^{1}_{n} = \rvev^{1}_{n} + (\bI_{n} +\bB_{n})\rvez^{1}_{n} + \rved^{1}_{n}.
\end{align}
In this case, 
\begin{align}
 I(\rvax_{1}^{n};\rvay_{1}^{n}) 
 &= I(\rvav_{1}^{M};\rvay_{1}^{n})
 =
 h(\rvey^{1}_{n}) 
 - 
 h(\rvey^{1}_{n}|\rvev^{1}_{n}) 
 \\&
 =
 h(\rvey^{1}_{n}) 
 - 
 h(( \bI_{n}+\bB_{n})\rvez^{1}_{n} + \rvev^{1}_{n} + \rved^{1}_{n}|\rvev^{1}_{M}) 
 \\&
 =
 h(\rvey^{1}_{n}) 
 - 
 h(( \bI_{n}+\bB_{n})\rvez^{1}_{n} + \rved^{1}_{n}|\rvev^{1}_{M})
 \\&
 =
 h(\rvey^{1}_{n}) 
 - 
 h(( \bI_{n}+\bB_{n})\rvez^{1}_{n}+ \rved^{1}_{n}) .
\end{align}
But $\lim_{n\to\infty}\frac{1}{n}( h(( \bI_{n}+\bB_{n})\rvez^{1}_{n} + \rvev^{1}_{n}+ \rved^{1}_{n})
 -
 h(( \bI_{n}+\bB_{n})\rvez^{1}_{n}+ \rved^{1}_{n}))=0,$
 which follows directly from applying Theorem~\ref{thm:eg_n_instate_w_disturb} to each of the differential entropies.
Notice that this result holds irrespective of how small the power of the disturbance may be.

Thus, the capacity-achieving scheme proposed in~\cite{kim-yh10} (and further studied in~\cite{ardfra12}), although of groundbreaking theoretical importance, would yield zero rate in any practical situation, since every real signal is unavoidably affected by some amount of noise.

\section{Conclusions}\label{sec:conclusions}
This paper has provided a geometrical insight and rigorous results for characterizing the increase in differential entropy rate (referred to as entropy gain) introduced by passing an input random sequence through a discrete-time linear time-invariant (LTI) filter $G(z)$ such that the first sample of its impulse response has unit magnitude.
Our time-domain analysis allowed us to explain and establish under what conditions the entropy gain coincides with what was predicted by Shannon, who followed a frequency-domain approach to a related problem in his seminal 1948 paper.
In particular, we demonstrated that the entropy gain arises only if $G(z)$ has zeros outside the unit circle (i.e., it is non-minimum phase, (NMP)).
This is not sufficient, nonetheless, since letting the input and output be $\rvau$ and $\rvay=G\rvau$, the difference $h(\rvay_{1}^{n})-h(\rvau_{1}^{n})$ is zero for all $n$, yielding no entropy gain.
However, if 
the distribution of the input process $\rvau$ satisfies a certain regularity condition (defined as being ``entropy balanced'') and the output 
has the form $\rvay=G\rvau + \rvaz$, with $\rvaz$ being an output disturbance with bounded differential entropy, we have shown that the entropy gain can range from zero to the sum of the logarithm of the magnitudes of the NMP zeros of $G(z)$, depending on how $\rvaz$ is distributed.
A similar result is obtained if, instead of an output disturbance, we let $G(z)$ have a random initial state.
We also considered the difference between the differential entropy rate of the \textit{entire} (and longer) output of $G(z)$ and that of its input, i.e., $h(\rvay_{1}^{n+\eta}) -h(\rvau_{1}^{n})$, where $\eta+1$ is the length of the impulse response of $G(z)$.
For this purpose, we introduced the notion of ``effective differential entropy'', which can be applied to a random sequence whose support has dimensionality smaller than its dimension.
Interestingly, the effective differential entropy gain in this case, which is intrinsic to $G(z)$, is also the sum of the logarithm of the magnitudes of the NMP zeros of $G(z)$, without the need to add disturbances or a random initial state.
We have illustrated some of the implications of these ideas in three problems.
Specifically, we used the fundamental results here obtained to provide a simpler and more general proof to characterize the rate-distortion function for Gaussian non-stationary sources and MSE distortion.
Then, we applied our results to provide sufficient conditions for equality in an information inequality of significant importance in networked control problems.
Finally, 
we showed that the information rate of the capacity-achieving scheme proposed in~\cite{kim-yh10} for the autoregressive Gaussian channel with feedback drops to zero in the presence of any additive disturbance in the channel input or output of sufficient (finite) length, no matter how small it may be.

\appendix
\subsection{Proofs of Results Stated in the Previous Sections}
\begin{proof}[Proof of Proposition~\ref{prop:gaussian_is_entropy_balanced}]
Let $\sigsq_{\rvau}$ be the per-sample variance of $\rvau_{1}^{\infty}$, thus 
$h(\rveu^{1}_{n})=\frac{n}{2}\log(2\pi\expo{}\sigsq_{\rvau})$.
Let
 $\rvey^{\nu+1}_{n} \eq \bPhi_{n}\rveu^{1}_{n}$.
 Then 
 $\bK_{\rvey^{\nu+1}_{n}}
 =
 \sigsq_{\rvau}\bPhi_{n}\bPhi_{n}^{T}
 =
 \sigsq_{\rvau}
 \bI_{n-\nu}$, where $\bI_{n-\nu}$ is the $(n-\nu)\times (n-\nu)$ identity matrix.
 As a consequence, $h(\rvey^{\nu+1}_{n})=([n-\nu]/2)\log(2\pi\expo{}\sigsq_{\rvau})$, and thus
 $\lim_{n\to\infty}\frac{1}{n}(h(\rvey^{\nu+1}_{n})-h(\rveu^{1}_{n}))=0$.
\end{proof}

\begin{proof}[Proof of Lemma~\ref{lem:piecewiseconstant}]
 Let $\set{b_{\ell}}_{\ell=1}^{\infty}$ be the intervals (bins) in $\Rl$ where the sample PDF is constant.
 Let $\set{p_{\ell}}_{\ell=1}^{\infty}$ be the probabilities of these bins.
 Define the discrete random process $\rvac_{1}^{\infty}$, where $\rvac(i)=\ell$ if and only if $\rvau_i\in b_{\ell}$.
Let $\rvey^{\nu+1}_{n}\eq \bPhi_{n}\rveu^{1}_{n}$ where $\bPhi_{n}\in\Rl^{(n-\nu)\times n}$ has orthonormal rows.
Then 
 \begin{align}\label{eq:volvera}
  h(\rvey^{\nu+1}_{n}) 
 & = 
  h(\rvey^{\nu+1}_{n}|\rvec^{1}_{n}) + I(\rvec^{1}_{n};\rvey^{\nu+1}_{n})
  \\&
   \leq
    h({\rvey}^{\nu+1}_{n}|\rvec^{1}_{n}) + I(\rvec^{1}_{n};\rveu^{1}_{n}) ,
 \end{align}
where the inequality is due to the fact that $\rveu^{1}_{n}$ and $\rvey^{\nu+1}_{n}$ are deterministic functions of $\rveu^{1}_{n}$, and hence
$
\rvec^{1}_{n}
\longleftrightarrow
\rveu^{1}_{n}
\longleftrightarrow
\rvey^{\nu+1}_{n}
$.
 Subtracting $h(\rveu^{1}_{n})$ from~\eqref{eq:volvera} we obtain
 \begin{align}
   h(\rvey^{\nu+1}_{n})
   -
    h(\rveu^{1}_{n})
    &\leq 
    h({\rvey}^{\nu+1}_{n}|\rvec^{1}_{n}) + I(\rvec^{1}_{n};\rveu^{1}_{n}) -    h(\rveu^{1}_{n})
    \\&
    =
    h({\rvey}^{\nu+1}_{n}|\rvec^{1}_{n}) - h(\rveu^{1}_{n}|\rvec^{1}_{n}).
 \end{align}
 Hence,
 \begin{align}
  \lim_{n\to\infty}\frac{1}{n}\left(h({\rvey}^{\nu+1}_{n}) -h(\rveu^{1}_{n})  \right)
  \leq
 \lim_{n\to\infty}\frac{1}{n}\left(h({\rvey}^{\nu+1}_{n}|\rvec^{1}_{n}) -h(\rveu^{1}_{n}|\rvec^{1}_{n})  \right)=0
 \end{align}
where the last equality follows from Lemma~\ref{lem:entropy_bounded_different_supports} (see Appendix~\ref{subsec:technical_lemmas})\, whose conditions are met because, given $\bc^{1}_{n}$, the sequence $\rveu^{1}_{n}$ has independent entries each of them distributed uniformly over a possibly different interval with bounded and positive measure.
The opposite inequality is obtained by following the same steps as in the proof of Lemma~\ref{lem:entropy_bounded_different_supports},
from~\eqref{eq:prior_to} onwards, which completes the proof.
\end{proof}

\begin{proof}[Proof of Lemma~\ref{lem:sum_yields_entropy_balanced}]
 Let $\rvey^{1}_{n}\eq [\bPsi_{n}^{T}|\bPhi_{n}^{T}]^{T}\rvew^{1}_{n}$, where $[\bPsi_{n}^{T}|\bPhi_{n}^{T}]^{T}\in\Rl^{n\times n}$ is a unitary matrix and where 
 $\bPsi_{n}\in\Rl^{\nu\times n}$ and $\bPhi_{n}\in\Rl^{(n-\nu)\times n}$ have orthonormal rows.
 Then
 \begin{align}
  h(\rvey^{\nu+1}_{n})
  & 
  =
  h(\rvey^{1}_{n})
  -
  h(\rvey^{1}_{\nu}|\rvey^{\nu+1}_{n})
\\& 
  =
  h(\rvew^{1}_{n})
  -
  h(\rvey^{1}_{\nu}|\rvey^{\nu+1}_{n})\label{eq:lalala}
  \end{align}
We can lower bound $h(\rvey^{1}_{\nu}|\rvey^{\nu+1}_{n})$ as follows:
  \begin{align}
   h(\rvey^{1}_{\nu}|\rvey^{\nu+1}_{n})
   &=
   h(\bPsi_{n}\rveu^{1}_{n} + \bPsi_{n}\rvev^{1}_{n}
   \,|\,
   \bPhi_{n}\rveu^{1}_{n} + \bPhi_{n}\rvev^{1}_{n})
   \\&
   \geq
   h(\bPsi_{n}\rveu^{1}_{n} + \bPsi_{n}\rvev^{1}_{n}
   \,|\,
   \bPhi_{n}\rveu^{1}_{n} + \bPhi_{n}\rvev^{1}_{n}\;,\; \rvev^1_{n})
\\&
   =
   h(\bPsi_{n}\rveu^{1}_{n} 
   \,|\,
   \bPhi_{n}\rveu^{1}_{n} + \bPhi_{n}\rvev^{1}_{n}\;,\; \rvev^1_{n})
\\&
   =
   h(\bPsi_{n}\rveu^{1}_{n} 
   \,|\,
   \bPhi_{n}\rveu^{1}_{n} ,\;\rvev^{1}_{n})
\\&
   =
   h(\bPsi_{n}\rveu^{1}_{n} 
   \,|\,
   \bPhi_{n}\rveu^{1}_{n}).
   \end{align}
Substituting this result into~\eqref{eq:lalala}, dividing by $n$ and taking the limit as $n\to\infty$, and recalling that, 
since $\rvau_{1}^{\infty}$ is entropy balanced, then $\lim_{n\to\infty}\frac{1}{n}h(\bPsi_{n}\rveu^{1}_{n} |   \bPhi_{n}\rveu^{1}_{n})=0$, lead us to $\lim_{n\to\infty}\frac{1}{n}(h(\bPhi_n\rvew^1_n)-h(\rvew^1_n)) \leq 0$. 

The opposite bound over $h(\rvey^{1}_{\nu}|\rvey^{\nu+1}_{n})$ can be obtained from
\begin{align}
h(\rvey^{1}_{\nu}|\rvey^{\nu+1}_{n})
&=
   h(\bPsi_{n}\rveu^{1}_{n} + \bPsi_{n}\rvev^{1}_{n}
   \,|\,
   \bPhi_{n}\rveu^{1}_{n} + \bPhi_{n}\rvev^{1}_{n})
\leq
   h(\bPsi_{n}\rveu^{1}_{n} + \bPsi_{n}\rvev^{1}_{n})
\leq h(\bPsi_{n}(\rvew_G)^1_n),
\end{align}
where $(\rvew_G)^1_n$ is a jointly Gaussian sequence with the same second-order moment as $\rvew^1_n$.
Therefore, $h(\bPsi_{n}(\rvew_G)^1_n) \leq \frac{\nu}{2} \log (2\pi\expo{} \max\set{\sigma^2_{\rvaw}(i)})$, with $\sigma^2_{\rvaw}(i)$ being the variance of the sample $\rvaw(i)$. 
The fact that $\rvew^1_n$ has a bounded second moment at each entry $\rvaw(i)$, and replacing the latter inequality in~\eqref{eq:lalala}, satisfy
$\lim_{n\to\infty}-\frac{1}{n} h(\rvey^1_\nu|\rvey^{\nu+1}_n) = \lim_{n\to\infty}\frac{1}{n}(h(\bPhi_n\rvew^1_n)-h(\rvew^1_n)) \geq 0$, which finishes the proof.

\end{proof}

\begin{proof}[Proof of Lemma~\ref{lem:filtering_preserves_entropy_balance}]
Let $\rvey^{1}_{n}\eq [\bPsi_{n}^{T} | \bPhi_{n}^{T}]^{T}\rvew^{1}_{n}$ where $[\bPsi_{n}^{T} | \bPhi_{n}^{T}]^{T}\in\Rl^{n\times n}$ is a unitary matrix and where 
 $\bPsi_{n}\in\Rl^{\nu\times n}$ and $\bPhi_{n}\in\Rl^{(n-\nu)\times n}$ have orthonormal rows.
Since $\rvew^{1}_{n} = \bG_{n}\rveu^{1}_{n}$, we have that
\begin{align}
 \bPsi_{n}\rvew^{1}_{n} = \bPsi_{n}\bG_{n}\rveu^{1}_{n}.
\end{align}
 Let $\bPsi_{n}\bG_{n}=\bA_{n}\Sigma_{n}\bB_{n}$ be the SVD of $\bPsi_n\bG_{n}$, where $\bA_{n}\in\Rl^{\nu\times\nu}$ is an orthogonal matrix, $\bB_{n}\in\Rl^{\nu\times n}$ has orthonormal rows and 
 $\Sigma_{n}\in\Rl^{\nu\times \nu}$ is a diagonal matrix with the singular values of $\bPsi_{n}\bG_{n}$.
 Hence
 \begin{align}
  h(\bPsi_{n}\rvew^{1}_{n}) 
  =
  h(\bPsi_{n}\bG_{n}\rveu^{1}_{n})
  =
  h(\bA_{n}\Sigma_{n}\bB_{n}\rveu^{1}_{n})
  =
  \log\det(\Sigma_{n}) +  h(\bB_{n}\rveu^{1}_{n}).
 \end{align}
 It is straightforward to show that the diagonal entries in $\Sigma_{n}$ are 
 lower and upper bounded by the smallest and largest singular values of $\bG_{n}$, say $\sigma_{\min}(n)$ and $\sigma_{\max}(n)$, respectively, which yields
 \begin{align}
  \nu\log\sigma_{\text{min}}(n) +  h(\bB_{n}\rveu^{1}_{n}) 
  \leq 
  h(\bPsi_{n}\rvew^{1}_{n})   
  \leq 
  \nu\log\sigma_{\text{max}}(n) +  h(\bB_{n}\rveu^{1}_{n}) .
 \end{align}
 But from Lemma~\ref{lem:singular_values_bounded}, 
 $\lim_{n\to\infty}(1/n)\sigma_{\text{min}}(n) =\lim_{n\to\infty}(1/n)\sigma_{\text{max}}(n)=0$, and thus
 \begin{align}
  \lim_{n\to\infty} \frac{1}{n}h(\bPsi_{n}\rvew^{1}_{n}) =     
  \lim_{n\to\infty} \frac{1}{n}h(\bB_{n}\rveu^{1}_{n}) =0,
 \end{align}
where the last equality is due to the fact that $\rvau_{1}^{\infty}$ is entropy balanced.
This completes the proof.
\end{proof}

\begin{proof}[Proof of Lemma~\ref{lem:singular_values_bounded}]
The fact that $\lim_{n\to\infty}\lambda_{n}(\bA_{n}\bA_{n}^{T})$ is upper bounded follows directly from the fact that $A(z)$ is a stable transfer function.
On the other hand,  $\bA_{n}$ is positive definite (with all its eigenvalues equal to $1$), and so $\bA_{n}\bA_{n}^{T}$ is positive definite as well, with $\lim_{n\to\infty}\lambda_{1}(\bA_{n}\bA_{n}^{T})\geq 0$.
Suppose that 
$\lim_{n\to\infty}\lambda_{1}(\bA_{n}\bA_{n}^{T})=0$.
If this were true, then it would hold that
$\lim_{n\to\infty}\lambda_{n}(\bA_{n}^{-1}\bA_{n}^{-T})=\infty$.
But $\bA_{n}^{-1}$ is the lower triangular Toeplitz matrix  associated with $A^{-1}(z)$, which is stable (since $A(z)$ is minimum phase), implying that 
$\lim_{n\to\infty}\lambda_{n}(\bA_{n}^{-1}\bA_{1}^{-T})<\infty$, thus leading to a contradiction.
This completes the proof.
\end{proof}

\begin{proof}[Proof of Lemma~\ref{lem:gap_with_two_terms}]\label{proof:lem_gap_with_two_terms}
Since $\bQ_{n}$ is unitary, we have that
\begin{align}\label{eq:hy_hw}
 h(\rvey^{1}_{n}) 
 = 
 h(\bQ_{n}\rvey^{1}_{n}) 
 = 
 h(
 \overbrace{
 \underbrace{\bD_{n}\bR_{n}\rveu^{1}_{n}}_{\rvev^{1}_{n}} 
 + 
 \underbrace{\bQ_{n}\rvez^{1}_{n}}_{\bar{\rvez}^{1}_{n}}
 }^{\rvew^{1}_{n}}) 
 = 
 h( \rvew^{1}_{n}),
\end{align}
where 
\begin{align}
\rvew^{1}_{n}&\eq \rvev^{1}_{n} + \bar{\rvez}^{1}_{n},\\
\rvev^{1}_{n}&\eq \bD_{n}\bR_{n}\rveu^{1}_{n}, \\
\bar{\rvez}^{1}_{n}&\eq \bQ_{n}\rvez^{1}_{n}.
\end{align}
Applying the chain rule of differential entropy, we get
\begin{align}\label{eq:hw_split}
 h(\rvaw_{1}^{n}) & = h(\rvaw_{1}^{m}) + h(\rvaw_{m+1}^{n}|\rvaw_{1}^{m}). 
\end{align}
Notice that $\rvew^{1}_{m} = [\bD_{n}]^{1}_{m}\bR_{n}\rveu^{1}_{n} + [\bQ_{n}]^{1}_{m}\rvez^{1}_{n}$.
Thus, it only remains to determine the limit of $h(\rvaw_{m+1}^{n}|\rvaw_{1}^{m})$ as $n\to\infty$.
We will do this by deriving a lower and an upper bound for this differential entropy and show that these bounds converge to the same expression as $n\to\infty$.

To lower bound $h(\rvaw_{m+1}^{n}|\rvaw_{1}^{m})$ we proceed as follows
\begin{align}
 h(\rvaw_{m+1}^{n}|\rvaw_{1}^{m})
 &
 =
 h(\rvav_{m+1}^{n}+\,\bar{\rvaz}_{m+1}^{n}|\rvav_{1}^{m}+\,\bar{\rvaz}_{1}^{m}) \label{eq:hw|w}
\\&
\overset{(a)}{\geq} 
h(\rvav_{m+1}^{n}+\,\bar{\rvaz}_{m+1}^{n}|\rvav_{1}^{m},\,\bar{\rvaz}_{1}^{m})
\\&
=
h(\rvav_{m+1}^{n}+\,\bar{\rvaz}_{m+1}^{n}  , \rvav_{1}^{m}|\,\bar{\rvaz}_{1}^{m})
-
h(\rvav_{1}^{m}|\,\bar{\rvaz}_{1}^{m})
\\&
\overset{(b)}{=}
h(\rvav_{m+1}^{n}+\,\bar{\rvaz}_{m+1}^{n} , \rvav_{1}^{m}|\,\bar{\rvaz}_{1}^{m})
-
h(\rvav_{1}^{m})
\\&
=
h(\rvav_{m+1}^{n}+\,\bar{\rvaz}_{m+1}^{n} |\,\bar{\rvaz}_{1}^{m})
+
h(\rvav_{1}^{m}|\,\bar{\rvaz}_{1}^{m}, \rvav_{m+1}^{n}+\,\bar{\rvaz}_{m+1}^{n})
-
h(\rvav_{1}^{m})
\\&
\overset{(c)}{\geq}  
h(\rvav_{m+1}^{n}  |\,\bar{\rvaz}_{1}^{m})
+
h(\rvav_{1}^{m}|\,\bar{\rvaz}_{1}^{m} , \rvav_{m+1}^{n}+\,\bar{\rvaz}_{m+1}^{n})
-
h(\rvav_{1}^{m})
\\&
\overset{(d)}{=}
h(\rvav_{m+1}^{n})
+
h(\rvav_{1}^{m}|\,\bar{\rvaz}_{1}^{m} , \rvav_{m+1}^{n}+\,\bar{\rvaz}_{m+1}^{n})
-
h(\rvav_{1}^{m})
\\&
\overset{(e)}{\geq}
h(\rvav_{m+1}^{n})
+
h(\rvav_{1}^{m}|\rvav_{m+1}^{n})
-
h(\rvav_{1}^{m})
\\&
= 
h(\rvav_{1}^{n})
-
h(\rvav_{1}^{m})
= 
h(\rvau_{1}^{n})
-
h(\rvav_{1}^{m}), \label{eq:h(w|w)=h(u)+h(v)}
\end{align}
where $(a)$ follows from including $\bar{\rvaz}_1^{m}$ (or $\rvav_1^m$ as well) to the conditioning set, while $(b)$ and $(d)$ stem from the independence between $\rvau_1^\infty$ and $\bar{\rvaz}_1^\infty$. Inequality $(c)$ is a consequence of $h(X+Y)\geq h(X)$, and $(e)$ follows from including $\bar{\rvaz}_{m+1}^n$ to the conditioning set in the second term, and noting that $h(\rvav_1^m)$ is not reduced upon the knowledge of $\rvaz_1^n$. 

On the other hand,
\begin{align}
 h(\rvav_{1}^{m}) 
 = 
 h([\bD_{n}]^{1}_{m}\bR_{n}\rveu^{1}_{n}) 
 = 
 \Sumfromto{i=1}{m}\log d_{n,i}
 +
  h([\bR_{n}]^{1}_{m}\rveu^{1}_{n}), \label{eq:arowana}
\end{align}
then, by inserting~\eqref{eq:arowana} and~\eqref{eq:h(w|w)=h(u)+h(v)} in~\eqref{eq:hw|w}, dividing by $n$, and taking the limit $n\to\infty$, we obtain
%
\begin{align}
 \lim_{n\to\infty}\frac{1}{n} 
 h(\rvaw_{m+1}^{n}|\rvaw_{1}^{m})
 &
 \geq
 \lim_{n\to\infty}\frac{1}{n}
 \left( h(\rvau_{1}^{n})- \Sumfromto{i=1}{m}\log d_{n,i}
 -
 h([\bR_{n}]^{1}_{m}\rveu^{1}_{n}) 
 \right)
 \\&
 =
 \bar{h}(\rvau_{1}^{\infty})
 -
 \lim_{n\to\infty}\frac{1}{n}\Sumfromto{i=1}{m}\log d_{n,i}, \label{eq:the_lower_bound}
\end{align}
where the last equality is a consequence of the fact that $\rvau_{1}^{\infty}$ is entropy balanced.

We now derive an upper bound for $h(\rvaw_{m+1}^{n}|\rvaw_{1}^{m})$.
Defining the random vector 
$$
\rvex^{m+1}_{n}\eq [\bR_{n}]^{m+1}_{n}\rveu^{1}_{n},
$$ 
we can write
\begin{align}
 [\bD_{n}]^{m+1}_{n} \bR_{n}\rveu^{1}_{n} =  \Dmn\rvex^{m+1}_{n}
\end{align}
where
\begin{align}
 \Dmn\eq \diag\set{d_{n,m+1},d_{n,m+2},\ldots,d_{n,n} }.
\end{align}
Therefore,
\begin{align}
h(\rvaw_{m+1}^{n}|\rvaw_{1}^{m})
 &
\leq
h(\rvew^{m+1}_{n})
=
h(\Dmn\rvex^{m+1}_{n} + \bar{\rvez}^{m+1}_{n})
\\&
=
\log\det(\Dmn)
+
h(\rvex^{m+1}_{n} + (\Dmn)^{-1}\bar{\rvez}^{m+1}_{n})
.\label{eq:agnus_dei}
\end{align}
Notice that by Assumption~\ref{assu:z}, $\bar{\rvez}^{m+1}_{n}=[\bQ_{n}]^{m+1}_{n}\rvez^{1}_{n} = [\bQ_{n}]^{m+1}_{n}[\bPhi]^{1}_{n}\rves^{1}_{\kappa}$, and thus is restricted to the span of $[\bQ_{n}]^{m+1}_{n}[\bPhi]^{1}_{n}$
of dimension $\kappa_{n}\leq \kappa$, for all $n\geq m+\kappa$.
Then, for $n>m+\kappa_{n}$, one can construct a unitary matrix $\bH_n\eq (\bA_{n}^{T} | \bB_{n}^{T})^{T}\in\Rl^{(n-m)\times(n-m)}$, 
such that the rows of $\bA_{n}\in\Rl^{\kappa \times (n-m)}$ span the space spanned by the columns of $(\Dmn)^{-1}[\bQ_{n}]^{m+1}_{n}[\bPhi]^{1}_{n}$ and such that  $\bB_{n}(\Dmn)^{-1}[\bQ_{n}]^{m+1}_{n}[\bPhi]^{1}_{n}=0$.
Therefore, from~\eqref{eq:agnus_dei},
\begin{align}
h(\rvaw_{m+1}^{n}|\rvaw_{1}^{m})
 &
\leq 
\log\det(\Dmn)
+
h(\bH_n\rvex^{m+1}_{n} + \bH_n(\Dmn)^{-1}\bar{\rvez}^{m+1}_{n})
\nonumber
\\&
= 
\log\det(\Dmn)
+
h(\bB_{n}\rvex^{m+1}_{n})
+
h(\bA_{n}\rvex^{m+1}_{n} + \bA_{n}(\Dmn)^{-1}\bar{\rvez}^{m+1}_{n}|\bB_{n}\rvex^{m+1}_{n})
\nonumber
\\&
\leq
\log\det(\Dmn)
+
h(\bB_{n}\rvex^{m+1}_{n})
+
h(\bA_{n}\rvex^{m+1}_{n} + \bA_{n}(\Dmn)^{-1}\bar{\rvez}^{m+1}_{n})
\nonumber
\\&
\leq 
\log\det(\Dmn)
+
h(\bB_{n}\rvex^{m+1}_{n})
+ \frac{1}{2}\log\left(2\pi\expo{} \det\left(\bK_{\bA_n\rvex^{m+1}_n} + \bK_{\bA_n({}^{m+1}							[\bD_n]_n)^{-1}\bar{\rvez}^{m+1}_n}\right)\right)
\nonumber
\\&
\leq 
\log\det(\Dmn)
+
h(\bB_{n}\rvex^{m+1}_{n})
\nonumber
+
\frac{1}{2}\log\left(2\pi\expo{} 
\left[
\lambda_{\max}(\bK_{\rvex^{m+1}_{n}})
+  
\frac{\lambda_{\max}(\bK_{\bar{\rvez}^{m+1}_{n}})}{\lambda_{\min}(\Dmn)^{2}}
\right]^{\kappa_{n}} \right)\label{eq:laultima}
\end{align}
where $\bK_{\bA_n\rvex^{m+1}_n}$ and $\bK_{\bA_n({}^{m+1}[\bD_n]_n)^{-1}\bar{\rvez}^{m+1}_n}$ are the covariance matrices of $\bA_n\rvex^{m+1}_n$ and $\bA_n({}^{m+1}[\bD_n]_n)^{-1}\bar{\rvez}^{m+1}_n$, respectively, and where the last inequality follows from~\cite{fiedle71}.
The fact that
$\lambda_{\max}(\bK_{\rvex^{m+1}_{n}})$ and 
$\lambda_{\max}(\bK_{\bar{\rvez}^{m+1}_{n}})$
are bounded and remain bounded away from zero for all $n$, 
and the fact that $\lambda_{\text{min}}(\Dmn)$ either grows with $n$ or decreases sub-exponentially (since the $m$ first singular values decay exponentially to zero, with $|\det\bD_n|=1$),
imply in~\eqref{eq:laultima} that 
\begin{align}
 \lim_{n\to\infty}\frac{1}{n}   
    h(\rvaw_{m+1}^{n}|\rvaw_{1}^{m})
\leq 
 \lim_{n\to\infty}\frac{1}{n}   
\log\det(\Dmn)
+
 \lim_{n\to\infty}\frac{1}{n}   
h(\bB_{n}\rvex^{m+1}_{n}).
\end{align}
But the fact that $\det\bD_{n}=1$ implies that $\log\det(\Dmn)=-\sumfromto{i=1}{m}\log d_{n,i}$.
This, together with the assumption that $\rvau_{1}^{\infty}$ is entropy balanced yields
\begin{align}
 \lim_{n\to\infty}\frac{1}{n}   
 h(\rvaw_{m+1}^{n}|\rvaw_{1}^{m})
\leq 
\bar{h}(\rvau_{1}^{\infty})
-
\lim_{n\to\infty}\frac{1}{n}   
\Sumfromto{i=1}{m}\log d_{n,i},
\end{align}
which coincides with the lower bound found in~\eqref{eq:the_lower_bound}, completing the proof.
\end{proof}

\begin{proof}[Proof of Lemma~\ref{lem:hashimoto_IIR}]\label{proof:lem_hashimoto_IIR}
 The transfer function $G(z)$ can be factored as
 $G(z) = \tilde{G}(z)F(z)$,
where $\tilde{G}(z)$ is stable and minimum phase and $F(z)$ is stable with all the non-minimum phase zeros of $G(z)$, both being biproper rational functions.
From Lemma~\ref{lem:singular_values_bounded}, in the limit as $n\to\infty$, the eigenvalues of $\tilde{\bG}_{n}^{T}\tilde{\bG}_{n}$ are lower and upper bounded by 
$\lambda_{\text{min}}(\tilde{\bG}^T\tilde{\bG})$
and 
$\lambda_{\text{max}}(\tilde{\bG}^T\tilde{\bG})$, 
respectively, where 
$0<\lambda_{\text{min}}(\tilde{\bG}^T\tilde{\bG})\leq \lambda_{\text{max}}(\tilde{\bG}^T\tilde{\bG})<\infty$.
Let 
$\tilde{\bG}_{n}=
\tilde{\bQ}_n^{T}\tilde{\bD}_{n}\tilde{\bR}_{n}$ 
and 
$\bF_{n}=
\bQ_n^{T}\bD_{n}\bR_{n}
$ 
be the SVDs of $\tilde{\bG}_{n}$ and $\bF_{n}$, respectively, with 
$\tilde{d}_{n,1}\leq \tilde{d}_{n,2}\leq \cdots \leq \tilde{d}_{n,n}$
and 
$d_{n,1}\leq d_{n,2}\leq \cdots \leq d_{n,n}$
being the diagonal entries of the diagonal matrices 
$\tilde{\bD}_{n}$,
$\bD_{n}$,
respectively.
Then 
\begin{align}
 \bG_{n}^{T}\bG_{n} 
 &
 =  \bF_{n}^{T}\tilde{\bG}_{n}^{T}\tilde{\bG}_{n}\bF_{n}
 =
(\tilde{\bD}_{n}\tilde{\bR}_{n} \bQ_n^{T}\bD_{n}\bR_{n})^{T}
 \tilde{\bD}_{n}\tilde{\bR}_{n} \bQ_n^{T}\bD_{n}\bR_{n}
\end{align}
Denoting the $i$-th row of $\bR_{n}$ by $\br_{n,i}^{T}$ be, we have that, from the Courant-Fischer theorem~\cite{horjoh85} that
\begin{align}
\lambda_{i}(\bG_{n}^{T}\bG_{n})
&\leq 
\max_{\bv \in \ospn\set{\br_{n,k}}_{k=1}^{i}\,:\, \norm{\bv}=1}
\norm{\bG\bv}^{2}
\\&
= 
\max_{\bv \in \ospn\set{\br_{n,k}}_{k=1}^{i}\,:\, \norm{\bv}=1}
\norm{
 \tilde{\bD}_{n}
  \tilde{\bR}_{n}^{T}\bQ_{n}^{T}\bD_{n} \bR_{n}
 \bv
}^{2}
\\&
\leq
d_{n,i}^{2}\tilde{d}_{n,n}^{2}
 \end{align}
Likewise,
\begin{align}
\lambda_{i}(\bG_{n}^{T}\bG_{n})
&\geq 
\min_{\bv \in \ospn\set{\br_{n,k}}_{k=i}^{n}\,:\, \norm{\bv}=1}
\norm{\bG\bv}
\\&
= 
\min_{\bv \in \ospn\set{\br_{n,k}}_{k=i}^{n}\,:\, \norm{\bv}=1}
\norm{
 \tilde{\bD}_{n}
  \tilde{\bR}_{n}^{T}\bQ_{n}^{T}\bD_{n} \bR_{n}
 \bv
}^{2}
\\&
\geq
d_{n,i}^{2}\tilde{d}_{n,1}^{2}
 \end{align}
Thus
\begin{align}
\lim_{n\to\infty}\frac{\lambda_{i}(\bG^{T}\bG)}{d_{n,i}^{2}}  
\in
 \left(\lambda_{\text{min}}(\tilde{\bG}^T\tilde{\bG})\ ,\ \lambda_{\text{max}}(\tilde{\bG}^T\tilde{\bG}) \right).
\end{align}
The result now follows directly from Lemma~\ref{lem:hashimoto} (in the appendix).
 \end{proof}

\begin{proof}[Proof of Theorem~\ref{thm:eg_n_instate_w_disturb_ineq} ]\label{proof:thm_eg_n_instate_w_disturb_ineq}
In this case 
\begin{align}
 [\bD_{n}]^{1}_{m}\bR_{n}\rveu^{1}_{n} +[\bQ_{n}]^{1}_{m}\rvez^{1}_{n}  
 &
 =
 [\bD_{n}]^{1}_{m}\bR_{n}\rveu^{1}_{n} +[\bQ_{n}]^{1}_{m} [\bPhi]^{1}_{n}\rves^{1}_{\kappa}.
\end{align}
Notice that the columns of the matrix $[\bQ_{n}]^{1}_{m} [\bPhi]^{1}_{n}\in\Rl^{m\times \kappa}$ span a space of dimension $\kappa_{n}\in\set{0,1,\ldots,\bar{\kappa}}$, which means that one can have 
$[\bQ_{n}]^{1}_{m} [\bPhi]^{1}_{n}=\bzero$ (if $\kappa_{n}=0$). 
In this case (i.e., if $\lim_{n\to\infty}[\bQ_{n}]^{1}_{m} [\bPhi]^{1}_{n}=\bzero$) then the lower bound is reached by inserting the latter expression into~\eqref{eq:gap_with_two_terms} and invoking Lemma~\ref{lem:hashimoto_IIR}.

We now consider the case in which $\lim_{n\to\infty}[\bQ_{n}]^{1}_{m} [\bPhi]^{1}_{n}\neq \bzero$.
This condition implies that there exists an $N$ sufficiently large such that 
$\kappa_{n}\geq 1$ for all $n\geq N$.
Then, for all $n\geq N$ there exist unitary matrices
\begin{align}
 \bH_{n}\eq 
 \begin{pmatrix}
  \bA_{n}
  \\
  \overline{\bA}_{n}
 \end{pmatrix}
 \in\Rl^{m\times m},\fspace n\geq N,
\end{align}
where 
$\bA_{n}\in\Rl^{\kappa_{n}\times m}$ 
and 
$\overline{\bA}_{n}\in\Rl^{(m-\kappa_{n})\times m}$ 
have orthonormal rows, such that
\begin{align}
 \bH_{n}[\bQ_{n}]^{1}_{m} [\bPhi]^{1}_{n}
 =
 \begin{pmatrix}
  \bA_{n}[\bQ_{n}]^{1}_{m} [\bPhi]^{1}_{n}
  \\
  \bzero
 \end{pmatrix}, \fspace n\geq N.
\end{align}
Thus
\begin{align}
 h\left( [\bD_{n}]^{1}_{m}\bR_{n}\rveu^{1}_{n} +[\bQ_{n}]^{1}_{m}\rvez^{1}_{n} \right) 
 &
 =
 h\left([\bD_{n}]^{1}_{m}\bR_{n}\rveu^{1}_{n} + [\bQ_{n}]^{1}_{m} [\bPhi]^{1}_{n}\rves^{1}_{\kappa}\right)
 \\&
 =h\left(\bH_{n}([\bD_{n}]^{1}_{m}\bR_{n}\rveu^{1}_{n} +[\bQ_{n}]^{1}_{m} [\bPhi]^{1}_{n}\rves^{1}_{\kappa})\right)
 \\&
 =
 h\left(
 \bA_{n}[\bD_{n}]^{1}_{m}\bR_{n}\rveu^{1}_{n} 
 +\bA_{n}[\bQ_{n}]^{1}_{m} [\bPhi]^{1}_{n}\rves^{1}_{\kappa} \, |\,  \overline{\bA}_{n}[\bD_{n}]^{1}_{m}\bR_{n}\rveu^{1}_{n}\right)\nonumber
 \\&
 \fspace\fspace\fspace\fspace\fspace\fspace\fspace\fspace\fspace\fspace\fspace\fspace
 \fspace\fspace\fspace\fspace
 +
 h(\overline{\bA}_{n}[\bD_{n}]^{1}_{m}\bR_{n}\rveu^{1}_{n})\label{eq:la100}.
 \end{align}
The first differential entropy on the RHS of the latter expression is uniformly upper-bounded because 
$\rvau_{1}^{\infty}$ is entropy balanced, $[\bD_{n}]^{1}_{m}$ has decaying entries, and $h(\rvas_1^{\kappa})<\infty$. 
%
%
%
For the last differential entropy, notice that
$
[\bD_{n}]^{1}_{m}\bR_{n}
=
\block[\bD_{n}]{1}{m} \rows[\bR_{n}]{1}{m}
$.
Consider the SVD 
$\overline{\bA}_{n}\block[\bD_{n}]{1}{m} \rows[\bR_{n}]{1}{m}
=\bV_{n}^{T}\bSigma_{n}\bW_{n}$, 
with
$\bV_{n}\in\Rl^{(m-\kappa_{n})\times (m-\kappa_{n})}$ being unitary,
$\bSigma_{n}\in\Rl^{(m-\kappa_{n})\times (m-\kappa_{n})}$ being diagonal, and
$\bW_{n}\in\Rl^{(m-\kappa_{n})\times n}$ having orthonormal rows.
We can then conclude that 
\begin{align}
 h(\overline{\bA}_{n}[\bD_{n}]^{1}_{m}\bR_{n}\rveu^{1}_{n}) =
 h(\bSigma_{n}\bW_{n}\rveu^{1}_{n})
 =
 \log\abs{\det(\bSigma_{n})} + h(\bW_{n}\rveu^{1}_{n}) .\label{eq:maldealturas}
\end{align}
Now, the fact that
\begin{align*}
 \overline{\bA}_{n}\block[\bD_{n}]{1}{m} \rows[\bR_{n}]{1}{m}(\overline{\bA}_{n}\block[\bD_{n}]{1}{m} \rows[\bR_{n}]{1}{m})^{T}
 =
  \overline{\bA}_{n}\block[\bD_{n}]{1}{m}\block[\bD_{n}]{1}{m} \overline{\bA}_{n}^{T}
  =
  \bV^{T}\bSigma\bW\bW^{T}\bSigma^{T}\bV
  =
  \bV^{T}\bSigma\bSigma^{T}\bV
\end{align*}
allows one to conclude that 
\begin{align}\label{eq:pasos}
\log\abs{\det{\bSigma\,}}= \frac{1}{2}\log|\det(\overline{\bA}_{n}(\block[\bD_{n}]{1}{m})^{2}\overline{\bA}_{n}^{T})|. 
\end{align}
Recalling that $\overline{\bA}_{n}=\rows[\bH_{n}]{\kappa_{n}+1}{m}$ and that $\bH_{n}\in\Rl^{m\times m}$ is unitary, 
it is easy to show (by using the Cauchy interlacing theorem~\cite{horjoh85}) 
that
\begin{align}
\frac{1}{2}\log\abs{\det(\overline{\bA}_{n}(\block[\bD_{n}]{1}{m})^{2}\overline{\bA}_{n}^{T})} 
\leq 
\Sumfromto{i=\kappa_{n}+1}{m}\log d_{n,i},
\end{align}
with equality achieved if and only if $\overline{\bA}_{n}=[\bzero\,|\, \bI_{m-\kappa_{n}}]$. 
Substituting this into~\eqref{eq:pasos} and then the latter into~\eqref{eq:maldealturas} we arrive to
\begin{align}
  h(\overline{\bA}_{n}[\bD_{n}]^{1}_{m}\bR_{n}\rveu^{1}_{n})
  \leq 
  h( [\bW_{n}]^{1}_{m}\rveu^{1}_{n})
  +
  \sum_{i=\kappa_{n}+1}^{m}\log d_{n,i}. 
\end{align}
Substituting this into~\eqref{eq:gap_with_two_terms}, exploiting the fact that $\rveu_{1}^{\infty}$ is entropy balanced and 
invoking Lemma~\ref{lem:hashimoto_IIR}
yields the upper bound in~\eqref{eq:eg_n_instate_w_disturb_ineq}.
Clearly, this upper bound is achieved if, for example,
$\rows[\bQ_{n}]{1}{\bar{\kappa}} \rows[\bPhi]{1}{n}(\rows[\bQ_{n}]{1}{\bar{\kappa}} \rows[\bPhi]{1}{n})^{T}$ is non-singular for all $n$ sufficiently large, since, in that case, $\kappa_{n}=\bar{\kappa}$ and
we can choose $\bA_{n}=[\bI_{\bar{\kappa}} \; \bzero]$
and $\overline{\bA}_{n}=[\bzero   \; \bI_{m-\bar{\kappa}}]$.
This completes the proof.
\end{proof}

\begin{proof}[Proof of Theorem~\ref{thm:eg_n_instate_w_disturb} ]\label{proof:thm:eg_n_instate_w_disturb}
 As in~\eqref{eq:G_Factorized_as_tilde_G_F}, the transfer function $G(z)$ can be factored as
 $G(z) = \tilde{G}(z)F(z)$,
where $\tilde{G}(z)$ is stable and minimum phase and $F(z)$ is a stable FIR transfer function with all the non-minimum-phase zeros of $G(z)$ ($m$ in total).
Letting $\tilde{\rveu}^{1}_{n}\eq \tilde{\bG}_{n}\rveu^{1}_{n}$, we have that 
$
 h(\rvey^{1}_{n}) 	= h(\bF_{n}\tilde{\rveu}^{1}_{n} + \rvez^{1}_{n})
$,
$
h(\tilde{\rveu}^{1}_{n}) 	= h(\rveu^{1}_{n})
$,
and that $\set{\tilde{\rvau}_{i}}_{i=1}^{\infty}$ is entropy balanced (from Lemma~\ref{lem:filtering_preserves_entropy_balance}).
Thus,
\begin{align}
h(\rvey^{1}_{n})-h(\rveu^{1}_{n})=
h(\bG_{n}\rveu^{1}_{n}+\rvez^{1}_{n} ) - h(\rveu^{1}_{n})
 =
 h(\bF_{n}\tilde{\rveu}^{1}_{n}+\rvez^{1}_{n} ) - h(\tilde{\rveu}^{1}_{n}).
\end{align}
This means that
the entropy gain of $\bG_{n}$ due to the output disturbance $\rvaz_{1}^{\infty}$ corresponds to the entropy gain of $\bF_{n}$ due to the same output disturbance.
One can then evaluate the entropy gain of $\bG_{n}$ by applying Theorem~\ref{thm:eg_n_instate_w_disturb_ineq} to the filter $F(z)$ instead of $G(z)$, which we do next. 

Since only the first $m$ values of $\rvaz_{1}^{\infty}$ are non zero, it follows that in this case $\bPhi=[\,\bI_{m}\,|\, \bzero\,]^{T}$ (see Assumption~\ref{assu:z}).
 Therefore, 
 $\det(\rows[\bQ_{n}]{1}{m}\rows[\bPhi]{1}{n}(\rows[\bQ_{n}]{1}{m}\rows[\bPhi]{1}{n})^{T})
 =
 \det(\block[\bQ_{n}]{1}{m}(\block[\bQ_{n}]{1}{m})^{T})$
and 
 the sufficient condition given in Theorem~\ref{thm:eg_n_instate_w_disturb_ineq} will be satisfied for $\kappa=m$ if  
 $\lim_{n\to\infty}|\det(\block[\bQ_{n}]{1}{m})|>0$,
where now $\bQ_{n}^{T}$ is the left unitary matrix in the SVD $\bF_{n}=\bQ_{n}^{T}\bD_{n}\bR_{n}$.
We will prove that this is the case by using a contradiction argument.
Thus, suppose the contrary, i.e., that 
\begin{align}\label{eq:falseclaim}
 \lim_{n\to\infty}\det \block{1}{m}=0.
\end{align}
Then, there exists a sequence of unit-norm vectors $\set{\bv_{n}}_{n=1}^{\infty}$, with $\bv_{n}\in\Rl^{m}$ for all $n$, such that 
\begin{align}\label{eq:labase}
 \lim_{n\to\infty}\norm{\bv_{n}^{T}\, \block{1}{m}} =0
\end{align}
For each $n\in\Nl$, define the $n$-length image vectors $\bt_{n}^{T}\eq \bv_{n}^{T}\rows{1}{m}$,  and decompose them as 
\begin{align}
 \bt_{n}=\begin{bmatrix}
          \balpha_{n}\\
          \bbeta_{n}
         \end{bmatrix}
\end{align}
such that $\balpha_{n}\in\Rl^{m}$ and $\bbeta_{n}\in\Rl^{n-m}$.
Then, from this definition and from~\eqref{eq:labase}, we have that
\begin{subequations}\label{subeq:limxn_and_limyn}
\begin{align}
\norm{\balpha_{n}}^{2} + \norm{\bbeta_{n}}^{2}&=1,\fspace \forall n\in\Nl,\\
 \lim_{n\to\infty} \norm{\balpha_{n}}&=0\\
 \lim_{n\to\infty} \norm{\bbeta_{n}}&=1
\end{align}
\end{subequations}
As a consequence,
\begin{align}
\norm{ \bF_{n}^{T}\bt_{n}} 
= 
\norm{ \bR_{n}^{T}\bD_{n}\bQ_{n}\bt_{n}} 
=
\norm{\bD_{n}\bQ_{n}\bt_{n}}
=
\norm{\block[\bD_{n}]{1}{m} \rows{1}{m}  \bt_{n}},
\end{align}
where the last equality follows from the fact that, by construction, $\bt_{n}^{T}$ is in the span of the first $m$ rows of $\bQ_{n}$, together with the fact that $\bQ_{n}$ is unitary (which implies that $\rows{m+1}{n}\bt_{n}=\bzero$).
Since the top $m$ entries in $\bD_{n}$ decay exponentially as $n$ increases, we have that 
\begin{align}\label{eq:Hash-Ari_bound}
 \norm{ \bF_{n}^{T}\bt_{n}} \leq  \mathcal{O}(\zeta_{n}|\rho_{M}|^{-n}),
\end{align}
where $\zeta_{n}$ is a finite-order polynomial of $n$ (from Lemma~\ref{lem:hashimoto}, in the Appendix).
But 
\begin{align}
  \norm{ \bF_{n}^{T}\bt_{n}}
  &=
  \Norm{ ([\bF_{n}]^{1}_{m})^{T}\balpha_{n} + ([\bF_{n}]^{m+1}_{n})^{T}\bbeta_{n}}
  \\&
  \geq 
\Norm{([\bF_{n}]^{m+1}_{n})^{T}\bbeta_{n}}
-
    \Norm{ ([\bF_{n}]^{1}_{m})^{T}\balpha_{n}} 
    \\&
    \geq 
   \sigma_{min}( ([\bF_{n}]^{m+1}_{n})^{T} ) \norm{\bbeta_{n}}
   -
   \sigma_{max}(([\bF_{n}]^{1}_{m})^{T})
   \Norm{ \balpha_{n}}
\end{align}
Taking the limit as $n\to\infty$,
\begin{align}
 \lim_{n\to\infty}
   \norm{ \bF_{n}^{T}\bt_{n}}
   &
   \geq 
   \left(\lim_{n\to\infty}\sigma_{\text{min}}( ([\bF_{n}]^{m+1}_{n})^{T} )\right) 
   \left(\lim_{n\to\infty} \norm{\bbeta_{n}}\right)
   -
   \sigma_{\text{max}}(([\bF_{n}]^{1}_{m})^{T})
   \left(\lim_{n\to\infty} \Norm{ \balpha_{n}}\right)
   \\&
   =
\lim_{n\to\infty}\sigma_\text{min}( ([\bF_{n}]^{m+1}_{n})^{T} )\label{eq:limitalone}
\end{align}
where we have applied~\eqref{subeq:limxn_and_limyn} and the fact that  
$\sigma_{\text{max}}(([\bF_{n}]^{1}_{m})^{T})$ is bounded and does not depend on $n$.
Now, notice that 
$
 \rows[\bF_{n}]{m+1}{n} (\rows[\bF_{n}]{m+1}{n})^{T}
$
is a Toeplitz matrix with the convolution of $f$ and $f^{-}$ (the impulse response of $F$ and its time-reversed version, respectively) on its first row and column.
It then follows from~\cite[Lemma~4.1]{gray--06} that 
\begin{align}
 \lim_{n\to\infty}\lambda_{\text{min}}(\rows[\bF_{n}]{m+1}{n}(\rows[\bF_{n}]{m+1}{n})^{T}) = 
 \min_{\w:\w\in\pipi} |F\ejw|^{2}
 >0
\end{align}
(the inequality is strict because all the zeros of $F(z)$ are strictly outside the unit disk).
Substituting this into~\eqref{eq:limitalone} we conclude that
\begin{align}
 \lim_{n\to\infty}\sigma_{\text{min}}( (\rows[\bF_{n}]{m+1}{n} )^{T} )
 >0,
\end{align}
which contradicts~\eqref{eq:Hash-Ari_bound}.
Therefore,~\eqref{eq:falseclaim} leads to a contradiction, completing the proof.
\end{proof}
\begin{proof}[Proof of Theorem~\ref{thm:RDF_non_stat}]\label{proof:RDF_non_stat}
Denote the Blaschke product~\cite{sebrgo97} of $A(z)$ as
\begin{align}\label{eq:B_def}
 B(z) \eq \frac{\prod_{i=1}^{M}(z-p_{i}) }{\prod_{i=1}^{M}p_{i}^{*}(z-1/p_{i}^{*})  },  
\end{align}
which clearly satisfies
\begin{align}
 |B\ejw| &= 1,\fspace\forallwinpipi \label{eq:magBis1}\\
 b_{0}&\eq \lim_{|z|\to\infty} B(z) = 
 \frac{1}{\prod_{i=1}^{M}p_{i}^{*}  },\label{eq:b0_of_B}
\end{align}
where $b_{0}$ is the first sample in the impulse response of $B(z)$.
Notice that~\eqref{eq:magBis1} implies that 
$
\lim_{n\to\infty}\frac{1}{n}\norm{\bB_{n}\rveu^{1}_{n}}^{2} 
=
\lim_{n\to\infty}\frac{1}{n}\norm{\rveu^{1}_{n}}^{2} 
$
for every sequence of random variables $\rvau_{1}^{\infty}$ with uniformly bounded variance.
Since $B(z)$ has only stable poles and its zeros coincide exactly with the poles  of $A(z)$, it follows that $B(z)A(z)$ is a stable transfer function.
Thus, the asymptotically stationary process $\tilde{\rvax}_{1}^{\infty}$ defined in~\eqref{eq:tildex_def} can be constructed as
\begin{align}
 \tilde{\rvex}^{1}_{n}\eq \bB_{n}\rvex^{1}_{n},
\end{align}
where $\bB_{n}$ is a Toeplitz lower triangular matrix with its main diagonal entries equal to $b_{0}$.

The fact that $B(z)$ is biproper with $b_{0}$ as in~\eqref{eq:b0_of_B} implies that for any $\rveu^{1}_{n}$ with finite differential entropy
\begin{align}\label{eq:hgain_B_is_neg}
 h(\bB_{n}\rveu^{1}_{n}) 
 = 
 h(\rveu^{1}_{n}) - n\underbrace{\sumfromto{i=1}{M}\log\abs{p_{i}}}_{\eq \Gsp},
\end{align}
which will be utilized next.

For any given $n\geq M$, suppose that $C(z)$ is chosen and $\rvex^{1}_{n}$ and $\rveu^{1}_{n}$ are distributed so as to minimize 
$I(\rvex^{1}_{n};\bC_{n}\rvex^{1}_{n}+\rveu^{1}_{n})$ subject to the constraint 
$\expe{\norm{\rvey^{1}_{n}-\rvex^{1}_{n}}^{2}} = \expe{\norm{(\bC_{n} -\bI)\rvex^{1}_{n}}^{2}}+ \expe{\norm{\rveu^{1}_{n}}}^{2}\leq D$ 
(i.e., $\rvex^{1}_{n},\rveu^{1}_{n}$ is a realization of  $R_{\rvax,n}(D)$), yielding the reconstruction %
\begin{align}
 \rvey^{1}_{n} = \bC_{n}\rvex^{1}_{n}+\rveu^{1}_{n}. 
\end{align}
Since we are considering mean-squared error distortion, it follows that, for rate-distortion optimality, $\rveu^{1}_{n}$ must be jointly Gaussian with $\rvex^{1}_{n}$.
From these vectors, define
\begin{align}
 \tilde{\rveu}^{1}_{n}&\eq \bB_{n}\rveu^{1}_{n},
\\
 \tilde{\rvey}^{1}_{n}&\eq \bB_{n}\rvey^{1}_{n} = \bB^{1}_{n}\bC_{n}\bB_{n}^{-1}\tilde{\rvex}^{1}_{n} +\tilde{\rveu}^{1}_{n},
 \\
 \bar{\rvey}^{1}_{n}& 
 \eq 
 \tilde{\rvey}^{1}_{n} + \rved^{1}_{n}
 =
 \bB^{1}_{n}\bC_{n}\bB_{n}^{-1}\tilde{\rvex}^{1}_{n} +\tilde{\rveu}^{1}_{n} +  \rved^{1}_{n}.  \label{eq:ybar_def}
 \end{align}
 where $\rved^{1}_{n}$ is a zero-mean Gaussian vector independent of $(\tilde{\rveu}^{1}_{n},\tilde{\rvex}^{1}_{n})$ with finite differential entropy such that
 $\rvad_{k}=0$, $\forall k > M$. 
 %
%
Then, we have that%
\footnote{The change of variables and the steps in this chain of equations is represented by the block diagrams shown in Fig.~\ref{fig:rdfnsproof}.}
\begin{align}
n R_{\rvax,n}(D)=
 I(\rvex^{1}_{n};\rvey^{1}_{n}) 
 & 
\overset{(a)}{=}
 I(\bB_{n}\rvex^{1}_{n}; \bB_{n} \rvey^{1}_{n})
  =
  I(\tilde{\rvex}^{1}_{n};\tilde{\rvey}^{1}_{n})
  \\&
=
  h(\tilde{\rvey}^{1}_{n}) 
  -
  h(\tilde{\rvey}^{1}_{n}|\tilde{\rvex}^{1}_{n})
  \\& 
   \overset{(b)}{=}
  h(\tilde{\rvey}^{1}_{n}) 
  -
  h(\tilde{\rveu}^{1}_{n}|\tilde{\rvex}^{1}_{n})
  \\&
   \overset{(c)}{=}
  h(\tilde{\rvey}^{1}_{n}) -
  h(\tilde{\rveu}^{1}_{n})
  \label{eq:equal_rates}
  \\&
   \overset{(d)}{=}
  h(\tilde{\rvey}^{1}_{n}) 
  -
  \left(
  h(\tilde{\rveu}^{1}_{n}+\rved^{1}_{n})
  +
  [h(\rveu^{1}_{n})-h(\tilde{\rveu}^{1}_{n}+\rved^{1}_{n})]
  -
  n\Gsp 
  \right)
  \\&
     \overset{(e)}{=}
  h(\tilde{\rvey}^{1}_{n}) 
  -
  h(\tilde{\rveu}^{1}_{n}+\rved^{1}_{n}|\tilde{\rvex}^{1}_{n})
  +
  n\Gsp 
  -[h(\rveu^{1}_{n})-h(\tilde{\rveu}^{1}_{n}+\rved^{1}_{n})]
\\&
     \overset{(f)}{=}
  h(\tilde{\rvey}^{1}_{n}) 
  -
  h(\bar{\rvey}^{1}_{n}|\bar{\rvex}^{1}_{n})
  +
  n\Gsp 
  -[h(\rveu^{1}_{n})-h(\tilde{\rveu}^{1}_{n}+\rved^{1}_{n})]
  \\&
    =
    h(\tilde{\rvey}^{1}_{n})
    -
    h(\bar{\rvey}^{1}_{n})
  +    
  I(\tilde{\rvex}^{1}_{n};\bar{\rvey}^{1}_{n})
  +
  n\Gsp 
  -
  [h(\rveu^{1}_{n})-h(\tilde{\rveu}^{1}_{n}+\rved^{1}_{n})]
  \\&
  \overset{\hphantom{(a)}}{\geq}
    I(\tilde{\rvex}^{1}_{n};\bar{\rvey}^{1}_{n})
  +
  n\Gsp 
  -
  [h(\rveu^{1}_{n})-h(\tilde{\rveu}^{1}_{n}+\rved^{1}_{n})],
  \end{align}
where $(a)$ follows from $\bB_{n}$ being invertible,
$(b)$ is due to the fact that 
$\tilde{\rvey}^{1}_{n}=\bC_{n}\tilde{\rvex}^{1}_{n} + \tilde{\rveu}^{1}_{n}$, 
$(c)$ holds because $\rveu^{1}_{n}\Perp \rvex^{1}_{n}$.
The equality
$(d)$ stems from 
$h(\tilde{\rveu}^{1}_{n})=h(\rveu^{1}_{n}) - n\Gsp$ (see~\eqref{eq:hgain_B_is_neg}). 
Equality holds in $(e)$ because $\tilde{\rvex}^{1}_{n} \Perp (\tilde{\rveu}^{1}_{n},\rved^{1}_{n})$ and
in $(f)$ because of~\eqref{eq:ybar_def}.
The last inequality holds because $\bar{\rvey}^{1}_{n}=\tilde{\rvey}^{1}_{n}+\rved^{1}_{n}$ and $\rved^{1}_{n}\Perp\tilde{\rvey}^{1}_{n}$.
But from Theorem~\ref{thm:eg_n_instate_w_disturb},
$\lim_{n\to\infty}\frac{1}{n}(h(\tilde{\rveu}^{1}_{n} +\rved^{1}_{m}) - h(\rveu^{1}_{n}))=0$, and thus
$R_{\rvax,n}(D)\geq  \lim_{n\to\infty}\frac{1}{n}(\tilde{\rvex}^{1}_{n};\bar{\rvey}^{1}_{n})+\Gsp $.

At the same time, the distortion for the source 
$\tilde{\rvex}^{1}_{n}$ 
when reconstructed as $\bar{\rvey}^{1}_{n}$ is 
\begin{align}
\lim_{n\to\infty} 
\frac{1}{n}
\Expe{\norm{\bar{\rvey}^{1}_{n} - \tilde{\rvex}^{1}_{n} }^{2}} 
&=
\lim_{n\to\infty} 
\frac{1}{n}
\left(
\Expe{\norm{\tilde{\rvey}-\tilde{\rvex}^{1}_{n}}^{2}} 
+
\Expe{\norm{\rved^{1}_{n}}^{2}} 
\right)
\overset{(a)}{=}
\lim_{n\to\infty} 
\frac{1}{n}
\Expe{\norm{\tilde{\rvey}-\tilde{\rvex}^{1}_{n}}^{2}} 
\\&
=
\lim_{n\to\infty} 
\frac{1}{n}
\Expe{\norm{\bB_{n}(\rvey^{1}_{n}-\rvex^{1}_{n})}^{2}}
\overset{(b)}{=}
\lim_{n\to\infty} 
\frac{1}{n}
\Expe{\norm{\rvey^{1}_{n} - \rvex^{1}_{n} }^{2}},
\end{align}
where $(a)$ holds because $\norm{\rved^{1}_{n}}=\norm{\rved^{1}_{M}}$ is bounded, and $(b)$ is due to the fact that, in the limit, $B(z)$ is a unitary operator.
Recalling the definitions of $R_{\tilde{\rvax}}(D)$ and $R_{\tilde{\rvax}}(D)$, we conclude that
$\lim_{n\to\infty}\frac{1}{n}(\tilde{\rvex}^{1}_{n};\bar{\rvey}^{1}_{n})\geq  R_{\tilde{\rvax},n}(D)$, and therefore
\begin{align}
 R_{\rvax}(D) - R_{\tilde{\rvax}}(D) \geq \sumfromto{i=1}{M}\log|p_{i}|.
\end{align}
In order to complete the proof, it suffices to show that 
$R_{\rvax}(D) - R_{\tilde{\rvax}}(D) \leq \sumfromto{i=1}{M}\log|p_{i}|$.
For this purpose, consider now the (asymptotically) stationary source $\tilde{\rvex}^{1}_{n}$,
and suppose that 
$\hat{\rvey}^{1}_{n} = \tilde{\rvex}^{1}_{n}+\rveu^{1}_{n}$ 
realizes $R_{\tilde{\rvax},n}(D)$.
Again, 
$\tilde{\rvex}^{1}_{n}$ and $\rveu^{1}_{n}$ will be jointly Gaussian, satisfying 
$\hat{\rvey}^{1}_{n} \Perp \rveu^{1}_{n}$ (the latter condition is required for minimum MSE optimality).
From this, one can propose an alternative realization in which the error sequence is 
$\tilde{\rveu}\eq \bB_{n}\rveu^{1}_{n}$, yielding an output 
$\tilde{\rvey}^{1}_{n}=\tilde{\rvex}^{1}_{n} + \tilde{\rveu}^{1}_{n}$ with 
$\tilde{\rvey}^{1}_{n} \Perp \tilde{\rveu}^{1}_{n}$.
Then
\begin{align}
 n R_{\tilde{\rvax},n}(D)
 =
 I(\tilde{\rvex}^{1}_{n};\hat{\rvey}^{1}_{n})
 &
 =
 h(\tilde{\rvex}^{1}_{n})
 -
  h(\tilde{\rvex}^{1}_{n}| \hat{\rvey}^{1}_{n})
  \\&
  \overset{(a)}{=}
   h(\tilde{\rvex}^{1}_{n})
 -
  h(\rveu^{1}_{n})
  \\&
    \overset{(b)}{=}
   h(\tilde{\rvex}^{1}_{n})
 -
  h(\tilde{\rveu}^{1}_{n})
  -
  n\Gsp
  \\&
    \overset{(c)}{=}
  h(\tilde{\rvex}^{1}_{n})
 -
  h(\tilde{\rveu}^{1}_{n} |\tilde{\rvey}^{1}_{n})
  -
  n\Gsp
  \\& 
    \overset{(d)}{=}
  h(\tilde{\rvex}^{1}_{n})
 -
  h(\tilde{\rvex}^{1}_{n} |\tilde{\rvey}^{1}_{n})
  -
  n\Gsp
  \\&
  \overset{\phantom{(a)}}{=}
  I(\tilde{\rvex}^{1}_{n};\tilde{\rvey}^{1}_{n})
  -  
  n\Gsp
  \\&
  \overset{\phantom{(a)}}{=}
  I(\bB_{n}\rvex^{1}_{n};\bB_{n}\rvey^{1}_{n})
    -  
  n\Gsp
   \\&
  \overset{(e)}{=}
  I(\rvex^{1}_{n};\rvey^{1}_{n})
  -  
  n\Gsp,
  \end{align}
  where $(a)$ follows by recalling that $\hat{\rvey}^{1}_{n} = \tilde{\rvex}^{1}_{n}+\rveu^{1}_{n}$  and because 
  $\hat{\rvey}^{1}_{n}  \Perp \rveu^{1}_{n}$,
  $(b)$ stems from~\eqref{eq:hgain_B_is_neg},
  $(c)$ is a consequence of $ \tilde{\rvey}^{1}_{n} \Perp \tilde{\rveu}^{1}_{n}$,
  $(d)$ follows from the fact that 
  $\tilde{\rvey}^{1}_{n} = \tilde{\rvex}^{1}_{n}+\tilde{\rveu}^{1}_{n}$.
  Finally, $(e)$ holds because $\bB_{n}$ is invertible for all $n$.
  Since, asymptotically as $n\to\infty$, the distortion yielded by $\rvey^{1}_{n}$ for the non-stationary source $\rvex^{1}_{n}$ is the same which is obtained when $\tilde{\rvex}^{1}_{n}$ is reconstructed as $\hat{\rvey}^{1}_{n}$ (recall~\eqref{eq:magBis1}), we conclude that $R_{\rvax}(D) - R_{\tilde{\rvax}}(D) \leq \sumfromto{i=1}{M}\log|p_{i}|$, completing the proof. 
\end{proof}

\subsection{Technical Lemmas}\label{subsec:technical_lemmas}
\begin{lem}\label{lem:entropy_bounded_different_supports}
Let $\rvau_{1}^{\infty}$ be a random process with independent elements, and where each element $\rvau_i$ is uniformly distributed over possible different intervals $[-\frac{a_i}{2},\frac{a_i}{2}]$, such that $a_{\text{max}}>|a_{i}|>a_{\text{min}}>0, \forall i\in \Nl$, for some positive and bounded $a_{min}<a_{max}$.
Then $\rvau_{1}^{\infty}$ is entropy balanced.
\finenunciado
\end{lem}

\begin{proof}
Without loss of generality, we can assume that $a_{i}>1$, for all $i$ (otherwise, we could scale the input by $1/a_{\text{min}}$, which would scale
the output by the same proportion,
increasing the  input entropy by $n \log(1/a_{\text{min}})$ and the output entropy by $(n-\nu)\log(1/a_{\text{min}})$,
without changing the result).
The input vector $\rveu^{1}_{n}$ is confined to an $n$-box $\Usp_n$ (the support of $\rvau_1^n$) of volume 
$\Vsp_n(\Usp_n) = \prod_{i=1}^{n}a_{i}$
and has entropy
$\log (\prod_{i=1}^{n}a_{i})$.
This support is an $n$-box which contains 
${n\choose k}2^{n-k}$ $k$-boxes of different $k$-volume.
Each of these $k$-boxes is determined by fixing $n-k$ entries in $\rveu^{1}_{n}$ to  $\pm a_{i}/2$, and letting the remaining $k$ entries sweep freely over $[-\frac{a_i}{2},\frac{a_i}{2}]$.
Thus, the $k$-volume of each $k$-box is the product of the $k$ support sizes $a_{i}$ of the associated  selected free-sweeping entries.
But recalling that $a_{i}>1$ for all $i$, the volume of each $k$-box can be upper bounded by $\prod_{i=1}^{n}a_{i}$.
With this, the added volume of all the $k$-boxes contained in the original $n$-box can be upper bounded as
\begin{align}
 \Vsp_{k}^{\Box} (\Usp_n)
 \leq 
 {n \choose k}2^{n-k}
 \prod_{i=1}^{n}a_{i}\label{eq:upperboundVk} .
\end{align}
We now use this result to upper bound the entropy rate of $\rvey^{\nu+1}_{n}$.

Let $\rvey^{1}_{n}\eq [\bPsi_{n}^{T} | \bPhi_{n}^{T}]^{T}\rveu^{1}_{n}$ where $[\bPsi_{n}^{T} |\bPhi_{n}^{T}]^{T}\in\Rl^{n\times n}$ is a unitary matrix and where 
 $\bPsi_{n}\in\Rl^{\nu\times n}$ and $\bPhi_{n}\in\Rl^{(n-\nu)\times n}$ have orthonormal rows.
From this definition, $\rvey^{\nu+1}_{n}$ will distribute over a finite region $\Ysp^{\nu+1}_n \subseteq \Rl^{n-\nu}$,
corresponding to the projection onto the $k$-dimensional span of the rows of $\bPhi_{n}$.
Hence, $h(\rvey^{\nu+1}_{n})$ is upper bounded by the entropy of a uniformly distributed vector over the same support, i.e., by 
$\log\Vsp_{n-\nu}(\Ysp^{\nu+1}_n)$, 
where
$\Vsp_{n-\nu}(\Ysp^{\nu+1}_n)$ is the $(n-\nu)$-dimensional volume of this support.
In turn, $\Vsp_{n-\nu}(\Ysp^{\nu+1}_n)$ is upper bounded by  the sum of the volume of all $(\nu-k)$-dimensional boxes contained in the $n$-box in which $\rveu^{1}_{n}$ is confined, which we already denoted by $\Vsp_{n-\nu}^\Box(\Usp_n)$, and which is upper bounded as in~\eqref{eq:upperboundVk}. 
Therefore,
\begin{align*}
 h(\rvey^{1+\nu}_{n}) &
 \leq  \log\Vsp_{n-\nu}(\Ysp^{\nu+1}_n)
 \leq 
 \log \Vsp_{n-\nu}^{\Box} ( \Usp_n )
 \leq 
\log\left( 
 \frac{n!}{(n-\nu)!\nu!}
2^{\nu}
 \prod_{i=1}^{n}a_{i} 
 \right)
 \\&
 =
 \log\left( 
 n^{\nu}
2^{\nu} 
 \right)
 +
 \log\left( 
 \frac{n!}{(n-\nu)!n^{\nu}\nu!}
 \right)
 +
 \log\left( \prod_{i=1}^{n}a_{i}\right).
\end{align*}
Dividing by $n$ and taking the limit as $n\to\infty$ yields
\begin{align}\label{eq:is_non_pos}
 \lim_{n\to\infty}\frac{1}{n}\left(h(\rvey_{n}^{\nu+1})-h(\rveu_{n}^{1}) \right)
 \leq 0 
\end{align}

 On the other hand,
 \begin{align}
 h(\rvey^{\nu+1}_{n}) 
 =
 h(\rvey^{1}_{n}) 
  -
 h(\rvey^{1}_{\nu}|\rvey^{\nu+1}_{n})
 \overset{(a)}{=}
  h(\rveu^{1}_{n})
 -
 h(\rvey^{1}_{\nu}|\rvey^{\nu+1}_{n})
 \geq 
   h(\rveu^{1}_{n})
 -
 h(\rvey^{1}_{\nu}),\label{eq:prior_to}
 \end{align}
 where $(a)$ follows because $[\bPsi_{n}^{T} | \bPhi_{n}^{T}]^{T}$ is an orthogonal matrix.
 Letting $(\rvey_{G})^{1}_{\nu}$ correspond to the jointly Gaussian sequence with the same second-order moments as $\rvey^{1}_{\nu}$, and recalling that the Gaussian distribution maximizes differential entropy for a given covariance, we obtain the upper bound
 \begin{align}
   h(\rvey^{1}_{\nu}) \leq h((\rvey_{G})^{1}_{\nu}) 
   \overset{(a)}{=} 
   \frac{1}{2}
   \log\left((2\pi\expo{})^\nu \det(\bPsi_{n}\diag\set{\sigsq_{\rvau_i}}_{i=1}^{n}\bPsi_{n}^{T} )\right) 
   \overset{(b)}{\leq} 
   \frac{\nu}{2} 
   \log\left(2\pi\expo{} \max\set{\sigsq_{\rvau_i}}_{i=1}^{n} \right), \label{eq:crece}
 \end{align}
where $(a)$ follows since the $\set{\rvau_i}_{i=1}^{n}$ are independent,
and $(b)$ stems from the fact that $\bPsi_{n}\in\Rl^{\nu\times n}$ has orthonormal rows and from the Courant-Fischer theorem~\cite{horjoh85}.
Since $\max\set{\sigsq_{\rvau_i}}_{i=1}^{n}$ is bounded for all $n$, we obtain by substituting~\eqref{eq:crece} into~\eqref{eq:prior_to} that 
$\lim_{n\to\infty} \frac{1}{n}(h(\rvey^{\nu+1}_{n})-h(\rveu^{1}_{n}) )\geq 0$.
The combination of this with~\eqref{eq:is_non_pos} yields
$\lim_{n\to\infty} \frac{1}{n}(h(\rvey^{\nu+1}_{n})-h(\rveu^{1}_{n}) )=0$, completing the proof. 

\end{proof}

We re-state here (for completeness and convenience) the unnumbered lemma in the proof of~{\cite[Theorem~1]{hasari80}} as follows:
\begin{lem}\label{lem:hashimoto}
Let the function $\iota$ be as defined in~\eqref{eq:iota_def} but for a transfer function $G(z)$ with no poles and having only a finite number of zeros, $m$ of which lie outside the unit circle.
Then,
\begin{align}
 \lambda_{l}(\bG_n\bG_n^T) 
 = 
 \begin{cases}  
 \alpha_{n,l}^{2}(\rho_{ \iota(l)})^{-2n}	&, \text{if } l\leq m,\\
 \alpha_{n,l}^{2}				&, \text{otherwise },
 \end{cases}
\end{align}
where the elements in the sequence $\set{\alpha_{n,l} }$ are positive and increase or decrease at most polynomially with $n$.
\finenunciado
\end{lem}
\begin{lem}\label{lem:initial_states}
 Let $P(z)=\frac{N(z)}{D(z)}$ be rational transfer function of order $p$ with relative degree 1, with initial state $\bx_{0}\in\Rl^{p}$.  
 Let $T(z)=\frac{\Gamma (z)}{\Theta(z)}$ be a biproper rational transfer function of order $t$  
 with initial state $\bs_{0}\in\Rl^{t}$.
 Let 
 \begin{align}
  y \eq u - P(z)T(z) y,
 \end{align}
where $u$ is an exogenous signal.
Then
\begin{align}
 y = D \cdot \frac{\Theta}{\Theta D +N \Gamma  }u,
\end{align}
where the initial state of $D(z)$ is $\bx_{0}$ and the initial state of 
$\Theta/(\Theta D +N \Gamma )$ can be taken to be $[\bx_{0}\;\;\bs_{0}]$.
\end{lem}

\begin{proof}
Let  
$D(z)=1-\sumfromto{i=1}{p}d_{i}z^{-i}$ and
$N(z) = \sumfromto{i=1}{p}n_{i}z^{-i}$.
Define the following variables:
\begin{align}
 x = \frac{1}{D}y, && w = N x, && s = \frac{1}{\Theta} w, && v = \Gamma s.
\end{align}
Then the recursion corresponding to $P(z)$ is
\begin{align}
 x_{k} &= \sumfromto{i=1}{p}d_{i}x_{k-i} + y_{k},\fspace k\geq 1,
 \\
w_{k} & = \sumfromto{i=1}{p}n_{i}x_{k-i},\fspace k\geq 1.
\end{align}
This reveals that the initial state of $P(z)$ corresponds to
\begin{align}
 \bx_{0}= [x_{1-p}\; \;x_{2-p}\;\; \cdots \; \;x_{0}].
\end{align}
Let 
$\Gamma(z)=\sumfromto{i=0}{t}\g_{i}z^{-i}$
and
$\Theta(z) = 1-\sumfromto{i=1}{t} \theta_{i}z^{-i}$.
Then $v = T(z) w$ can be written as
\begin{align}
s_{k} &= \sumfromto{i=1}{t}\theta_{i}s_{k-i} +  w_{k},
\\ 
v_{k} &= \sumfromto{i=0}{t}\g_{i}s_{k-i},\fspace k\geq 1,
\end{align}
which reveals that the initial state of $T(z)$ can be taken to be 
\begin{align}
 \bs_{0}\eq [s_{1-t} \;\; s_{2-t}\;\;\cdots \;\; s_{0} ].
\end{align}
Since $y_{k}=u_{k}-v_{k}$, it follows that 
\begin{align}
 x_{k}
& 
=
\sumfromto{i=1}{p}d_{i}x_{k-i}  -v_{k} +u_{k}
=
\sumfromto{i=1}{p}d_{i}x_{k-i}  
- 
\sumfromto{i=1}{t}\g_{i}s_{k-i}
+
u_{k},\fspace k\geq 1.
\end{align}
Combining the above recursions, it is found that $y$ is related to the input $u$ by the following recursion:
\begin{align}
 x_{k} & = \sumfromto{i=1}{p}d_{i}x_{k-i}  - \sumfromto{i=1}{t}\g_{i}s_{k-i}
+
u_{k},\fspace k\geq 1,\\
 s_{k} & = \sumfromto{i=1}{t}\theta_{i}s_{k-i} + \sumfromto{i=1}{p}n_{i}x_{k-i},\fspace k\geq 1,\\
 y_{k} & = x_{k} -\sumfromto{i=1}{p}d_{i}x_{k-i},\fspace k\geq 1,
\end{align}
which corresponds to 
\begin{align}
 y = 
 \underbrace{D}_{\text{init. state $\bx_{0}$}}
 \cdot
 \overbrace{
 \underbrace{\frac{\Theta}{\Theta D+N\Gamma}}_{\text{init. state $[\bx_{0},\bs_{0}]$}}
 u}^{x}.
\end{align}

\end{proof}

%

\end{document}